\documentclass[11pt]{amsart}
\reversemarginpar
\usepackage{amsmath}
\usepackage{amssymb}
\usepackage{graphicx}
\usepackage{graphicx}
\usepackage{dcolumn}
\usepackage{bm}
\usepackage{amsfonts}
\usepackage{latexsym}
\usepackage{color}

\begin{document}
\newcommand{\commentout}[1]{}

\newcommand{\nwc}{\newcommand}
\newcommand{\bz}{{\mathbf z}}
\newcommand{\sqk}{\sqrt{\ks}}
\newcommand{\sqkone}{\sqrt{|\k\alpha|}}
\newcommand{\sqktwo}{\sqrt{|\k\beta|}}
\newcommand{\invsqkone}{|\k\alpha|^{-1/2}}
\newcommand{\invsqktwo}{|\k\beta|^{-1/2}}
\newcommand{\partz}{\frac{\partial}{\partial z}}
\newcommand{\grady}{\nabla_{\by}}
\newcommand{\gradp}{\nabla_{\bp}}
\newcommand{\gradx}{\nabla_{\bx}}
\newcommand{\invf}{\cF^{-1}_2}
\newcommand{\myphi}{\hat\br_{(\eta,\rho)}}
\newcommand{\minrg}{|\min{(\rho,\phi^{-1})}|}
\newcommand{\al}{\alpha}
\newcommand{\xvec}{\vec{\mathbf x}}
\newcommand{\kvec}{{\vec{\mathbf k}}}
\newcommand{\lt}{\left}
\newcommand{\ksq}{\sqrt{\ks}}
\newcommand{\rt}{\right}
\newcommand{\ga}{\phi}
\newcommand{\vas}{\varepsilon}
\newcommand{\lan}{\left\langle}
\newcommand{\ran}{\right\rangle}
\newcommand{\tvas}{{W_z^\vas}}
\newcommand{\psiep}{{W_z^\vas}}
\newcommand{\wep}{{W^\vas}}
\newcommand{\weptil}{{\tilde{W}^\vas}}
\newcommand{\wepz}{{W_z^\vas}}
\newcommand{\weps}{{W_s^\ep}}
\newcommand{\wepsp}{{W_s^{\ep'}}}
\newcommand{\wepzp}{{W_z^{\vas'}}}
\newcommand{\wepztil}{{\tilde{W}_z^\vas}}
\newcommand{\vvas}{{\tilde{\ml L}_z^\vas}}
\newcommand{\veptil}{{\tilde{\ml L}_z^\vas}}
\newcommand{\vep}{{{ V}_z^\vas}}
\newcommand{\cvc}{{{\ml L}^{\ep*}_z}}
\newcommand{\cvcp}{{{\ml L}^{\ep*'}_z}}
\newcommand{\cvp}{{{\ml L}^{\ep*'}_z}}
\newcommand{\cvtil}{{\tilde{\ml L}^{\ep*}_z}}
\newcommand{\cvtilp}{{\tilde{\ml L}^{\ep*'}_z}}
\newcommand{\vtil}{{\tilde{V}^\ep_z}}
\newcommand{\ktil}{\tilde{K}}
\newcommand{\n}{\nabla}
\newcommand{\tkappa}{\tilde\kappa}
\newcommand{\Om}{{\Omega}}
\newcommand{\bx}{\mb x}
\nwc{\bv}{\mb v}
\newcommand{\br}{\mb r}
\nwc{\bH}{{\mb H}}
\newcommand{\bu}{\mathbf u}
\nwc{\bxp}{{{\mathbf x}}}
\nwc{\byp}{{{\mathbf y}}}
\newcommand{\bD}{\mathbf D}
\nwc{\bS}{\mathbf S}
\newcommand{\bA}{\mathbf \Phi}
\nwc{\bPhi}{\mathbf \Phi}
\nwc{\bd}{{\mathbf d}}
\nwc{\cO}{\mathcal{O}}
\nwc{\co}{\mathcal{o}}
\nwc{\bG}{{\mathbf G}}
\nwc{\bF}{{\mathbf F}}
\nwc{\bR}{{\mathbf R}}
\nwc{\bh}{\mathbf h}
\newcommand{\bB}{\mathbf B}
\newcommand{\bC}{\mathbf C}
\newcommand{\bp}{\mathbf p}
\newcommand{\bq}{\mathbf q}
\newcommand{\by}{\mathbf y}
\nwc{\bI}{\mathbf I}
\nwc{\bP}{\mathbf P}
\nwc{\bs}{\mathbf s}
\nwc{\bX}{\mathbf X}
\newcommand{\pdg}{\bp\cdot\nabla}
\newcommand{\pdgx}{\bp\cdot\nabla_\bx}
\newcommand{\one}{1\hspace{-4.4pt}1}
\newcommand{\corr}{r_{\eta,\rho}}
\newcommand{\rinf}{r_{\eta,\infty}}
\newcommand{\rzero}{r_{0,\rho}}
\newcommand{\rzeroinf}{r_{0,\infty}}
\nwc{\om}{\omega}
\nwc{\thetatil}{{\tilde\theta}}

\nwc{\nwt}{\newtheorem}
\nwc{\xp}{{x^{\perp}}}
\nwc{\yp}{{y^{\perp}}}
\nwt{remark}{Remark}
\nwt{corollary} {Corollary}
\nwt{definition}{Definition} 

\nwc{\ba}{{\mb a}}
\nwc{\bal}{\begin{align}}
\nwc{\be}{\begin{equation}}
\nwc{\ben}{\begin{equation*}}
\nwc{\bea}{\begin{eqnarray}}
\nwc{\beq}{\begin{eqnarray}}
\nwc{\bean}{\begin{eqnarray*}}
\nwc{\beqn}{\begin{eqnarray*}}
\nwc{\beqast}{\begin{eqnarray*}}

\nwc{\eal}{\end{align}}
\nwc{\ee}{\end{equation}}
\nwc{\een}{\end{equation*}}
\nwc{\eea}{\end{eqnarray}}
\nwc{\eeq}{\end{eqnarray}}
\nwc{\eean}{\end{eqnarray*}}
\nwc{\eeqn}{\end{eqnarray*}}
\nwc{\eeqast}{\end{eqnarray*}}

\nwc{\ep}{\varepsilon}
\nwc{\eps}{\varepsilon}
\nwc{\ept}{\ep }
\nwc{\vrho}{\varrho}
\nwc{\orho}{\bar\varrho}
\nwc{\ou}{\bar u}
\nwc{\vpsi}{\varpsi}
\nwc{\lamb}{\lambda}
\nwc{\Var}{{\rm Var}}

\nwt{proposition}{Proposition}
\nwt{theorem}{Theorem}
\nwt{summary}{Summary}
\nwc{\nn}{\nonumber}
\nwc{\mf}{\mathbf}
\nwc{\mb}{\mathbf}
\nwc{\ml}{\mathcal}

\nwc{\IA}{\mathbb{A}} 
\nwc{\bi}{\mathbf i}
\nwc{\bo}{\mathbf o}
\nwc{\IB}{\mathbb{B}}
\nwc{\IC}{\mathbb{C}} 
\nwc{\ID}{\mathbb{D}} 
\nwc{\IM}{\mathbb{M}} 
\nwc{\IP}{\mathbb{P}} 
\nwc{\II}{\mathbb{I}} 
\nwc{\IE}{\mathbb{E}} 
\nwc{\IF}{\mathbb{F}} 
\nwc{\IG}{\mathbb{G}} 
\nwc{\IN}{\mathbb{N}} 
\nwc{\IQ}{\mathbb{Q}} 
\nwc{\IR}{\mathbb{R}} 
\nwc{\IT}{\mathbb{T}} 
\nwc{\IZ}{\mathbb{Z}} 
\nwc{\pdfi}{{f^{\rm i}}}
\nwc{\pdfs}{{f^{\rm s}}}
\nwc{\pdfii}{{f_1^{\rm i}}}
\nwc{\pdfsi}{{f_1^{\rm s}}}
\nwc{\chis}{{\chi^{\rm s}}}
\nwc{\chii}{{\chi^{\rm i}}}
\nwc{\cE}{{\ml E}}
\nwc{\cP}{{\ml P}}
\nwc{\cQ}{{\ml Q}}
\nwc{\cL}{{\ml L}}
\nwc{\cX}{{\ml X}}
\nwc{\cW}{{\ml W}}
\nwc{\cZ}{{\ml Z}}
\nwc{\cR}{{\ml R}}
\nwc{\cV}{{\ml V}}
\nwc{\cT}{{\ml T}}
\nwc{\crV}{{\ml L}_{(\delta,\rho)}}
\nwc{\cC}{{\ml C}}
\nwc{\cA}{{\ml A}}
\nwc{\cK}{{\ml K}}
\nwc{\cB}{{\ml B}}
\nwc{\cD}{{\ml D}}
\nwc{\cF}{{\ml F}}
\nwc{\cS}{{\ml S}}
\nwc{\cM}{{\ml M}}
\nwc{\cG}{{\ml G}}
\nwc{\cH}{{\ml H}}
\nwc{\bk}{{\mb k}}
\nwc{\cbz}{\overline{\cB}_z}
\nwc{\supp}{{\hbox{\rm supp}(\theta)}}
\nwc{\fR}{\mathfrak{R}}
\nwc{\bY}{\mathbf Y}
\newcommand{\mbr}{\mb r}
\nwc{\pft}{\cF^{-1}_2}
\nwc{\bU}{{\mb U}}

\title{Compressive Inverse Scattering I. High Frequency SIMO/MIMO Measurements}
\author{Albert C.  Fannjiang}
\email{
fannjiang@math.ucdavis.edu}

       \address{
   Department of Mathematics,
    University of California, Davis, CA 95616-8633}
   
       \begin{abstract}
Inverse scattering from discrete  targets
with the single-input-multiple-output (SIMO), multiple-input-single-output (MISO) or 
multiple-input-multiple-output (MIMO) measurements is analyzed
 by
compressed sensing theory with and without the Born 
approximation. 

High frequency analysis of  (probabilistic) recoverability by the $L^1$-based minimization/regularization principles  is presented. In the absence of noise, it is shown
that  the $L^1$-based solution  can recover exactly the target of sparsity up to the dimension of the data  either with 
the MIMO measurement
for the Born scattering  or  with the SIMO/MISO measurement for the exact scattering. 
The stability with respect to noisy data is proved
for weak or widely separated scatterers. Reciprocity between
the SIMO and MISO measurements is analyzed.  Finally  a
coherence bound (and the resulting  recoverability)  is proved
for 
diffraction tomography with high-frequency, few-view and limited-angle  SIMO/MISO measurements.

       \end{abstract}
       
       \maketitle
       
     \section{Introduction}
          
      
  \begin{figure}[t]
\begin{center}
\includegraphics[width=0.5\textwidth]{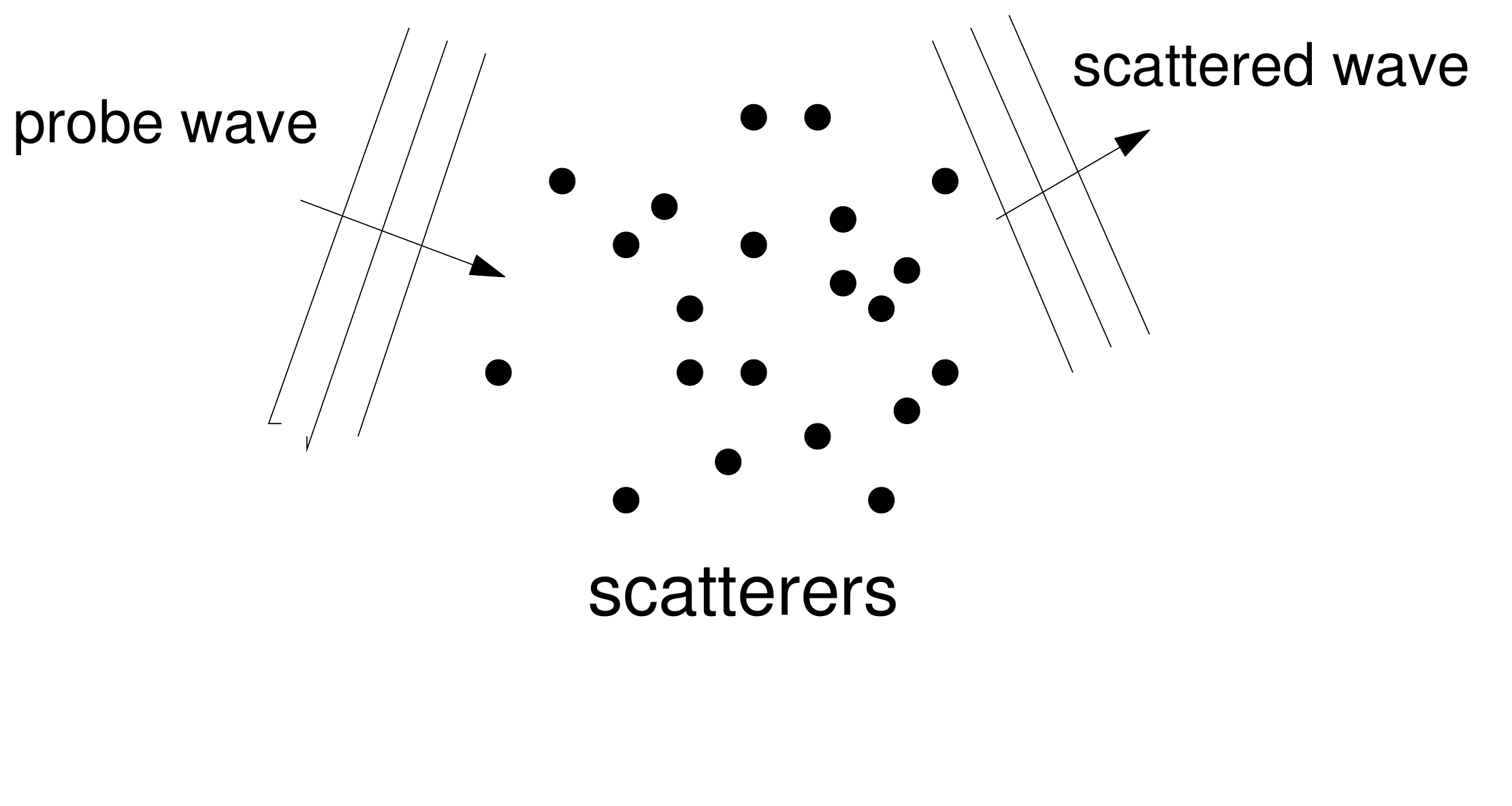}
\end{center}
\caption{Far-field  imaging geometry}
\end{figure}
\commentout{
\begin{figure}[t]
\begin{center}
\includegraphics[width=0.45\textwidth]{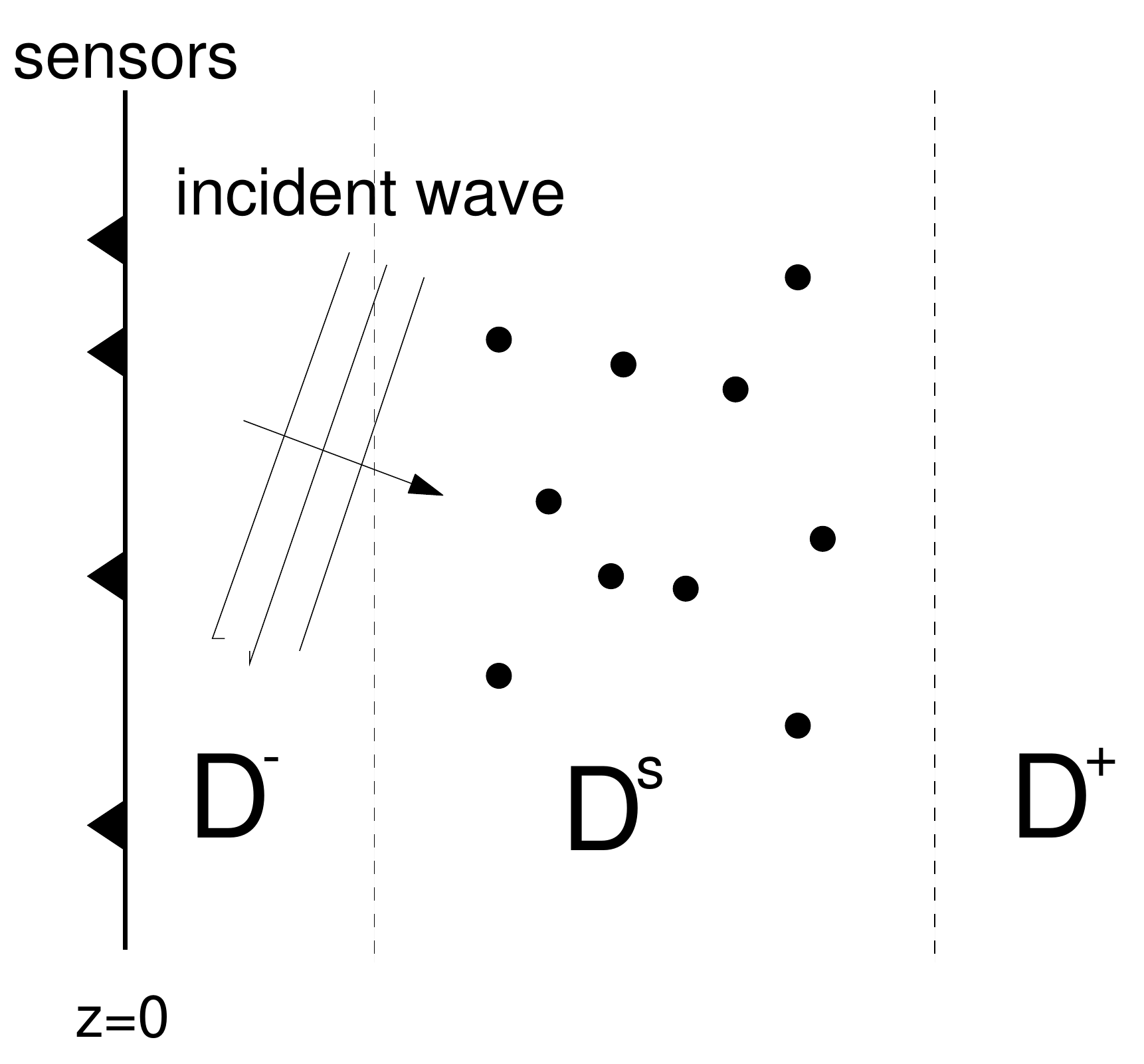}
\end{center}
\caption{Near-field imaging geometry}
\label{fig-dt}
\end{figure}
}
     
A monochromatic  wave $u$
propagating in a heterogeneous medium characterized by
a variable refractive index $n=\sqrt{1+\nu}$ is governed by the  Helmholtz
  equation 
  \beq
  \label{helm}
  \Delta u(\br)+\om^2(1+\nu(\br)) u(\br)=0
  \eeq 
  where $\nu\in \IC$ describes the medium inhomogeneities. 
 For simplicity,   the wave velocity is assumed to be unity and 
 hence the wavenumber $\om$ equals the frequency.

 Consider the plane wave incidence
 \beq
 \label{inc}
 u^{\rm i}(\br)=e^{i\om\br\cdot \bd}
 \eeq
 where $\bd\in S^{d-1}, d=2,3, $ is the incident direction. 
The scattered field $u^{\rm s}=u-u^{\rm i}$
then satisfies 
\beq
\label{scattered}
\Delta u^{\rm s}+\om^2 u^{\rm s}=-\om^2\nu u 
\eeq
which can be written as the Lippmann-Schwinger equation:
\beq
\label{exact'}
u^{\rm s}(\br)&=&\om^2\int_{\IR^d} \nu(\br') 
\lt(u^{\rm i}(\br')+u^{\rm s}(\br')\rt) G(\br, \br')d\br' 
\eeq
where $G$ is  the Green function for the operator $-(\Delta+\om^2)$ (see Appendix \ref{recipro} for the reciprocal
formulation).

The scattered field has
the far-field asymptotic \cite{Mel}
\beq
u^{\rm s}(\br)={e^{i\om |\br|}\over |\br|^{(d-1)/2}}\lt(A
(\hat\br,\bd)+\cO\lt({1\over |\br| }\rt)\rt),\quad\hat\br=\br/|\br|,\quad
d=2,3
\eeq
where the scattering amplitude $A$ is  determined by the formula  \cite{RS}
\beq
\label{sa}
A(\hat\br,\bd)&=&{\om^2\over 4\pi}
\int_{\IR^d}  \nu(\br') u(\br') e^{-i\om\br'\cdot\hat\br}d\br'. 
\eeq
In the inverse scattering theory, the scattering amplitude is the observable data and 
the main objective  then is to reconstruct
 $\nu$ from the knowledge
of the scattering amplitude.  In this paper,  we  use the $L^1$-minimization principle  called the {\em Basis Pursuit} to
study the inverse scattering from 
 point  scatterers. Note that
since $u$ in (\ref{sa}) is part of the unknown, 
the inverse scattering problem  is nonlinear. Physically speaking,
the nonlinearity is the consequence of multiple scattering
among the different scatterers.

The  standard   
theory of  inverse scattering asserts   the injectivity  of
the mapping from  $\nu\in C^1_c$ with a nonnegative imaginary part    to 
the corresponding scattering amplitude  
for a fixed frequency in {\em three}  (or higher) dimensions (Theorem 5.5 of
\cite{CCM}. See also \cite{HN, Mel, Nov2, Nov3, Ram} for similar results, \cite{KV, Nac, Nac2, SU} for inverse boundary-value problem and
 \cite{CK,  LP, Maj, Mel} for inverse obstacle scattering
). 
 In this case, the refractive index  can be determined uniquely by 
the full knowledge of $A(\hat\br,\bd), \forall \bd,\hat\br$, for a fixed $\om$.    Indeed, as 
$A$ is analytic in both $\bd$ and $\hat\br$, it suffices
to know $A$ for a countably many incident  and sampling directions 
in order to determine $\nu$ uniquely. As far as we know, the
uniqueness in two dimensions with a fixed frequency is still an open question. What is known for two dimensions is that
the uniqueness  holds if the scattering amplitude is given  for an interval of frequencies \cite{CCM}.

Of obvious theoretical interest, the uniqueness result by itself 
is of limited practical interest. All existing methods for determining the refractive index without linearizing 
the problem are based on the constrained nonlinear optimization in the $L^2$-norm for which
the exact recoverability is usually difficult to establish, especially
in the case of undersampling \cite{CCM, CK}. In this paper we show that a target is the unique, global
minimizer of an optimization principle
based on the $L^1$-norm if the target satisfies certain sparsity
constraint. Moreover, this $L^1$-minimization problem
can be effectively solved by linear programming as well
as various low-complexity greedy algorithms. 

\commentout{
such as the following one 
\beq
\min\sum_{j=1}^p\|A(\cdot, \bd_j)-
\cF u_j\|_2^2+\|u_j-u^{\rm i}_j-\om^2 \cL_\nu u_j\|^2_{L^2(D)}
\eeq
where $\cF$ is the far-field operator and
\[
\cL_\nu u_j(\br)=\int \nu(\br') u_j(\br')G(\br,\br')d\br'
\]
}

\commentout{
\begin{proposition} 
In three (or higher)  dimensions, the map from
real-valued function $\nu\in C^\infty_c$ to the scattering matrix is injective
\end{proposition}
}

 \commentout{
  A classical result on uniqueness  of  inverse scattering due to Schiffer \cite{LP} reads like this.
  
  \begin{proposition} 
  Assume that $D_1$ and $D_2$ are two sound-soft scatterers
  such that their far field patterns coincide for an infinite number
  of incident plane waves with distinct directions
  and one fixed frequency. Then $D_1=D_2$.
  \end{proposition}
  If the obstacle is contained in a ball, then finitely many incident
  plane waves and the resulting far-field patterns suffice to
  determine the support of the scatterer \cite{CS, Gin}.
   A long standing problem is if the far-field pattern for one single incident direction with one fixed frequency determines
  a sound-soft scatterer without any additional {\em a priori}
  information \cite{CK}. 
  Some progress has been obtained 
  regarding uniqueness with one or two  incident plane waves
  for polyhedral scatterers  \cite{AR, CY, EY, LZ}. 
}

 In this paper we focus on  the 
 two dimensional setting ($\br=(x,z)\in \IR^2$)
 for the aforementioned reason as well as 
 the notational simplicity. Although the details
 of the results are dimension-dependent,
our approach is not limited to higher dimensions.  
We discuss the three dimensional case  briefly  in Sections \ref{sec:3} and \ref{sec:near2}. 

Consider the medium with point scatterers located in 
 a square lattice 
$
\cL=\lt\{\br_i=(x_i,z_i): i=1,...,m\rt\}
$
of spacing  $\ell$. The total number $m$ of grid points
in $\cL$  is  a perfect square. 
Without loss of generality, assume $x_j=j_1\ell, z_j=j_2\ell $
where $j=(j_1-1)\sqrt{m}+j_2$ and $j_1, j_2=1,...,\sqrt{m}$.
Let  $\nu_{j}, j=1,...,m$ be the strength of
the scatterers.  Let  $
\cS=\lt\{\br_{i_j}=(x_{i_j}, z_{i_j}): j=1,...,s\rt\}$
be the
locations of the scatterers. Hence $\nu_j=0, \forall \br_j\not \in \cS$. When there is no risk of confusion,  we shall write
$\nu=(\nu_j)$ in the sequel. 

The scattering amplitude
for this medium is 
 a finite sum 
\beq
\label{10.11}
A(\hat\br,\bd)&=&{\om^2\over 4\pi}
\sum_{j=1}^m \nu_{j} u(\br_{j})
 e^{-i\om \br_j\cdot\hat\br}. 
\eeq
Moreover,  in analogy to (\ref{exact'}), the exciting field $u(\br_{i_j})$ satisfies
the Foldy-Lax equation \cite{TKDA, Mis} 
\beq
\label{fl}
u(\br_{i_l})&=&u^{\rm i}(\br_{i_l})+\om^2\sum_{l\neq j}G(\br_{i_l},\br_{i_j}) \nu_{i_j}
u(\br_{i_j}),\quad l=1,...,s
 \eeq
where all the  multiple scattering effects are included but the  self field is excluded to  avoid blow-up. 
\commentout{
In this paper we consider the far-field measurement
geometry as in Figure \ref{fig-dt} where the recorded data consist of the amplitudes
of the scattered plane waves in various  directions.
In the near-field measurement, the recorded data 
consist of the scattered field at the randomly selected
positions (see Figure \ref{fig-dt}).  
In inverse scattering the measured data 
are  the homogeneous  wave amplitudes
in the direction of $\tilde\theta_l, l=1,...,n$
\beq
\label{10.11}
u^{\rm s}(\tilde\theta_l)=\sum_{j=1}^s \nu_{i_j} u(\br_{i_j})
e^{i\om(\beta_l z_{i_j}-\alpha_l x_{i_j})},\quad \alpha_l=\cos{\tilde\theta_l},\quad\beta_l=\sin{\tilde\theta_l}. 
\eeq
}

\section{Methods and  results}
\subsection{MIMO  Born scattering}
First consider the Born approximation to (\ref{10.11}).  In   the Born approximation (also known as
   Rayleigh-Gans scattering in optics), the exciting
   field $u(\br_{i_j})$  is replaced 
   by the incident field $u^{\rm i}(\br_{i_j})$. 
   This  approximation  linearizes the relation between the 
  the target strength and the scattering amplitude and
  is valid for  sufficiently weak or
   widely separated scatterers. 
  
  For  the Born scattering, we define the target vector  $X=\nu \in \IC^m$ and use 
    multiple incident  waves 
  \beq
  \label{120}
 u^{\rm i}_k(\br)=e^{i\om (z\sin{\theta_k}+x\cos{\theta_k})},
 \quad k=1,...,p 
 \eeq
 where  $\theta_k$ is the incident  angle of the $k$-th probe wave. 
Throughout the paper we consider the single-input-multiple-output (SIMO), multiple-input-single-output (MISO) and  the multiple-input-multiple-output (MIMO) measurements  in which 
for each incident angle $\theta_k$  the resulting scattering amplitude is
measured at the multiple  sampling angles $\tilde\theta_l, l=1,...,n$.
After normalization  by $\om^2/(4\pi)$, the totality of the collected data forms 
the   measurement vector $Y\in \IC^{pn}$.  The corresponding sensing matrix in the linear relationship 
  $Y=\bPhi X$ has  the $(n(k-1)+l,j)$-entry \beq
\label{14.3}
e^{-i\om(z_j\sin{\tilde\theta_l}+x_j\cos{\tilde\theta_l})}u^{\rm i}_k (\br_j).
\eeq
where $\tilde\theta_l$ is the sampling angle of the $l$-th sensor.

Recent breakthrough in compressed
sensing has established  the insight  that the target can be recovered exactly with nearly minimum 
sensing resources by 
the $L^1$-minimization principle, called  basis pursuit (BP)\begin{equation}
\min \|X\|_1 \qquad \text{s.t.} \,\, \bPhi X=Y  
\label{L1}
\end{equation}
if the target is sufficiently sparse and
 the matrix $\bPhi$ satisfies either  the incoherence property
or the restricted isometry property \cite{BDE, CRT1, CRT2, CT, CT2,  CDS}. 
The $L^1$-minimization problem (\ref{L1})  can be solved by linear  programming \cite{BV, CT, CDS} or by  various  greedy
algorithms \cite{DM, NTV, Tro}.

In this paper we adopt  the incoherence approach 
to analyzing the SIMO/MISO and MIMO inverse scattering problems. 
Previously we have shown in \cite{subwave-cs}
that  suitably designed SIMO and MIMO measurements 
with {\em  planar} domains  (in three dimensions)
produce  random partial  Fourier matrix as the
sensing matrix which possesses a nearly optimal
 restricted isometry constant (RIC) with respect to the
sparsity of the target. Current  compressed sensing theory
predicts that  (\ref{L1}) yields a  superior performance \cite{Can, CT2}. However, with  non-planar domains 
the SIMO/MIMO measurements seem  to produce a rather poor  RIC.
Hence we adopt the alternative approach of incoherence in this paper. 
The restricted isometry approach will be taken
up in Part II for multi-shot SISO (single-input-single-output)
measurements  with non-planar domains for which
the framework of  random partial Fourier matrix can
be restored by special sampling schemes. 

To the best of our knowledge the present paper is the first rigorous study of  inverse scattering, including multiple scattering, in the framework of compressed sensing. Earlier studies \cite{Car2, Her,  Mar3, Mar2, PL, YL} of related problems
largely  take the compressed sensing theory for granted
and assume (explicitly or implicitly) either the incoherence or
 the restricted isometry property without proof.

Let us state the perhaps simplest criterion for exact  recoverability of the incoherence approach. 
\begin{proposition} \cite{DE,GN}
\label{thm:spark}
BP reconstructs perfectly any target $X$ of sparsity 
\beq
\label{spark3}
s\leq  {1\over 2} \lt({1\over \mu(\bPhi)}+1\rt)
\eeq
where the sparsity  $s=\|X\|_0$ is
the number of nonzero components in $X$ and
the coherence parameter $\mu(\bPhi)$ is defined as
\[
\mu(\bPhi)=\max_{i\neq j} {\lt|\sum_{l} \Phi_{li}\Phi^*_{lj}\rt|\over
\sqrt{\sum_{l}|\Phi_{li}|^2\sum_{l} |\Phi_{lj}|^2}} . 
\]
\end{proposition}
Proposition \ref{thm:spark} implies that the lower  the
coherence of the sensing matrix is  the more massive the exactly recoverable
target can be.  Under the condition (\ref{spark3}), a simple
greedy algorithm called Orthogonal Matching Pursuit (OMP)
can provably find the minimizer of (\ref{L1}) in at most
$s$ iterations \cite{Tro}. 

To construct sensing matrices of  low coherence
let us define the MIMO-sensor ensemble as follows.  Let the incident  angles $\theta_k, k=1,...,p,$ be  independently and identically  distributed according to
the probability density function  $\pdfi\in C^h([-\pi,\pi])$;
for every  incident angle 
 let the sampling  angles $\thetatil_l, l=1,...,n,$ be  independently and identically  distributed according to the probability density function $\pdfs\in C^h([-\pi,\pi])$ where $h>0$ is the degree of smoothness. 
Define $\hbox{\rm supp}(\pdfi)=\{\theta: \pdfi(\theta)\neq 0\}$
and $\hbox{\rm supp}(\pdfs)=\{\theta: \pdfs(\theta)\neq 0\}$. 
\commentout{
The main assumption on  $\pdfi$ and $\pdfs$ is 
that $\hbox{\rm supp}(\pdfi)$ and $\hbox{\rm supp}(\pdfs)$
are a finite union of disjoint open intervals such that
they vanish at the end points in each open interval $(a,b)$ 
faster than power $d$, i.e. 
\beq
\lim_{\theta\to a^+} {\pdfi (\theta)\over (\theta-a)^h}
=\lim_{\theta\to a^+} {\pdfs (\theta)\over (\theta-a)^h}=0,&&
\lim_{\theta\to b^-} {\pdfi (\theta)\over (\theta-b)^h}
=\lim_{\theta\to b^-} {\pdfs (\theta)\over (\theta-b)^h}=0.
\eeq
}
We call $\theta_*\in [-\pi, \pi]$
a {\em Blind Spot}  
  if there exists a pair
$\br,\br'\in\cL$ such that 
\beq
\label{bs}
\lt|{(\br-\br' )}\cdot (\cos\theta_*,\sin{\theta_*})\rt|=
|\br-\br'|.\eeq
In other words, the set of Blind Spots consists
of all the angles between the $x-$axis and $\br-\br', \forall
\br,\br'\in\cL$. 

In Section \ref{sec:born}, we prove that the following coherence
bound for the sensing matrix with entries (\ref{14.3}).

 \begin{theorem}\label{thm4}
 Let the sensing matrix $\bPhi$ be given according to the sensor ensemble. 
Suppose
\beq
\label{m-2}
m\leq {\delta\over 8} e^{K^2/2},\quad \delta, K>0.
\eeq
Then   the sensing matrix (\ref{14.3}) satisfies the coherence bound
\beq
\label{mut}
\mu(\bPhi) < \lt(
\chi^{\rm i}+{\sqrt{2} K\over \sqrt{p}}\rt)
\lt(\chi^{\rm s}+{\sqrt{2}K\over \sqrt{n}}\rt)
\eeq
 with probability greater than $(1-\delta)^2$
 where in general $\chii$ (resp. $\chis$) satisfies the bound
 \beq
 \label{21-3}
&\chi^{\rm i}\leq {c_t}{ {(1+\om \ell)}^{-1/2}} \|\pdfi\|_{t,\infty},&\\
\hbox{resp.}  &\chi^{\rm s}\leq {c_t}{(1+\om \ell)^{-1/2}} \|\pdfs\|_{t,\infty},&\label{21-4}
 \eeq
 where $\|\cdot\|_{t,\infty}$ is the H\"older norm
 of order $t>1/2$ and the constant $c_t$  depends only on $t$.
 If, however,  $\hbox{\rm supp}(\pdfi)$ (resp. $\hbox{\rm supp}(\pdfs)$) does not
 contains any Blind Spot, then $\chii$ (resp. $\chis$)  satisfies the bound
\beq
\label{21-1}
&\chi^{\rm i}\leq {c_h} (1+\om \ell)^{-h} \|\pdfi\|_{h,\infty},& 
\\
\hbox{resp.}  &\chi^{\rm s}\leq c_h(1+\om \ell)^{-h} \|\pdfs\|_{h,\infty},&  \|\pdfi\|_{h,\infty}=\sum_{|k|\leq h} \lt\|{d^k\over d\theta^k}\pdfi\rt\|_\infty \label{21-2}
 \eeq
 where the constant $c_h$ depends only on  $h$. 

\commentout{
 If, in addition to (\ref{m-1}), the scatterers lie in  a line, then 
\beq
\mu(\bPhi)< {2K^2 \over\sqrt{np}}+{c\over {(\om \ell)}^{3/2}}
\lt({\sqrt{2} K\over \sqrt{p}}+{\sqrt{2}K\over \sqrt{n}}\rt)+{c^2\over( \om \ell)^{3/2}} 
\eeq
for some constant $c>0$ with probability greater than $(1-\delta)^2$. 
\label{thm3'-1}
}
\label{thm3-1}
\end{theorem}

\begin{remark}
Theorem \ref{thm4} along with  Proposition \ref{thm:spark} then imply that any target of
sparsity up to 
\beq
\label{spark4}
s\leq {1\over 2}+{1\over 2}\lt(
\chi^{\rm i}+{\sqrt{2} K\over \sqrt{p}}\rt)^{-1}
\lt(\chi^{\rm s}+{\sqrt{2}K\over \sqrt{n}}\rt)^{-1}
\eeq
 can be exactly
recovered by BP. 

If $\om\ell$ is sufficiently large
(which is  the {high frequency} limit
referred to in the title), the dominant term on the right hand side
of (\ref{spark4}) is 
\beq
\label{hf}
{\sqrt{np}\over 4K^2}
\eeq
in view of (\ref{21-3})-(\ref{21-2}). 

Note that the high frequency limit for the Helmholtz equation
is different from that for the Schr\"odigner equation. The high frequency quantum  scattering is essentially linear (without
the Born approximation)
and can be solved by the Radon transform \cite{Mel}. 
\end{remark}

\commentout{
\begin{remark}
\label{rmk:three}
In three dimensions, $\chi^{\rm i}$ and $\chi^{\rm s}$ satisfy
different bounds. Indeed, the extra dimension results in
a faster decay of coherence. This is briefly discussed
in Section \ref{sec:three}.  For brevity, we state here the result for arbitrary distributions $\pdfi, \pdfs\in C^1$: Instead of
(\ref{21-3})-(\ref{21-4}) we have
\beq
&\chi^{\rm i}\leq {c}(1+\om \ell)^{-1}\|\pdfi\|_{1,\infty}&\\
\hbox{resp.}  &\chi^{\rm s}\leq {c}{(1+\om \ell)^{-1}}\|\pdfs\|_{1,\infty}.&
\eeq
Consequently, the asymptotic behavior (\ref{hf})
sets in faster in three dimensions than in two dimension in general. 
\end{remark}
}

To improve the sparsity constraint (\ref{spark3}),
Tropp \cite{Tro2} develops an approach
 in which the recoverability is only 
probabilistic in the following ensemble of targets. 
Let  the {\em target ensemble} consist of target vectors  with at most 
 $s$ non-zero entries  whose
 phases are 
 independently uniformly distributed in $[0,2\pi]$  and
 whose support  indices  are independently and  randomly selected from the index set $\{1,2,...,m\}$. 

The following theorem is a reformulation of results due to Tropp \cite{Tro2}. 
We refer the reader to \cite{cs-par} for the derivation of Proposition \ref{tropp}. 

 \begin{proposition}
 \label{tropp} Assume the matrix $\bPhi$ has all unit columns.
 Let $X$ be drawn from the target ensemble. 
 Assume that  
 \beq
 \label{M}
 \mu^2(\bPhi) s\leq \lt( 8\ln{{m\over \tau }}\rt)^{-1},\quad
 \tau\in (0,1) 
 \eeq
 and that for $q\geq 1$
 \beq
 \label{Op}
 3\lt({q\ln{s}\over 2\ln{{m\over \tau}}}\rt)^{1/2}+{s\over m}\|\bPhi\|_2^2\leq {1\over 4 e^{1/4}}.
 \eeq
 Then  $X$ is the unique solution of BP  with probability $1-2\tau -s^{-q}$. 
 Here $\|\bPhi\|_2$ denotes the spectral norm of $\bPhi$. 
 \end{proposition}
 \begin{remark}
 When the matrix $\bPhi$ has all unit { elements}, as in (\ref{14.3}), 
 the condition (\ref{Op}) becomes
 \beq
 \label{Op'}
  3\lt({q\ln{s}\over 2\ln{{m\over \tau }}}\rt)^{1/2}+{s\over m\rho}\|\bPhi\|_2^2\leq {1\over 4 e^{1/4}}
 \eeq
 where $\rho=\#\hbox{rows in}\,\, \bPhi$. 
 \end{remark}
Proposition \ref{tropp}  calls for 
the control of the spectral norm of $\bPhi$,
in addition to $\mu(\bPhi)$,  
in order to relax the sparsity constraint from (\ref{spark3})
to (\ref{M}). 

In Section \ref{sec:spec} we prove the following
spectral norm bound. 
\begin{theorem}
\label{thm5-2}
Under the  assumptions of Theorem \ref{thm4} the matrix $\bPhi$ has full rank
and its spectral norm satisfies the bound
\beq
\label{norm}
\|\bPhi\|_2^2\leq {2m}
\eeq
with probability greater than
\beq
\label{23-3}
\lt(1-c_1 \sqrt{np-1\over m}\rt)^{n(n-1)p(p-1)},
\quad n, p\geq 2
\eeq
for some constant $c_1>0$. 

In the SIMO  case $p=1$
we have \beq
\label{norm2}
\|\bPhi\|_2^2\leq {2m}
\eeq
with probability larger than 
\beq
\label{23-9}
\lt(1-c_1\sqrt{n-1\over m}\rt)^{n(n-1)},
\quad n\geq 2. 
\eeq

\end{theorem}
\begin{remark}
\label{rmk2}
The probability bounds (\ref{23-3}) and (\ref{23-9})  are  probably far from being  optimal. For $np\ll m$, a lower bound 
for (\ref{23-3}) would be
\[
1-c_1n(n-1)p(p-1)\sqrt{np-1\over m}
\]
which requires $m\gg (np)^5$ to be close to unity. 

When this is violated, we have to rely on Theorem \ref{thm4}
and Proposition \ref{thm:spark} which together 
guarantees recovery with probability greater than $(1-\delta)^2$ but with a higher sensor-to-target ratio. 

\end{remark}

Now we are ready to prove the main result for inverse Born
scattering. 
\begin{theorem}
\label{thm1}
Let the sensors and the target  be drawn randomly from the sensor and target ensembles, respectively,
and consider the sensing matrix $\bPhi$ 
of  the entries defined by  (\ref{120})-(\ref{14.3}).
If  (\ref{m-2}) holds, 
then  the targets of sparsity up to 
\beq
\label{spark}
s<\lt(8\ln{m\over \ep} \rt)^{-1}\lt(
\chi^{\rm i}+{\sqrt{2} K\over \sqrt{p}}\rt)^{-2}
\lt(\chi^{\rm s}+{\sqrt{2}K\over \sqrt{n}}\rt)^{-2}
\eeq
can be recovered by BP  with probability greater than 
\beq
\label{180}
\lt(\lt(1-c_1\sqrt{np-1\over m}\rt)^{n(n-1)p(p-1)}-2\delta\rt)  \lt(1-2\tau-s^{-q}\rt),\quad n\geq 2,\quad p\geq 2
\eeq
for some constant $c_1>0$ where, for $pn\gg s$,  $q$ can be chosen as
\[
 q={\ln {m}-\ln{\ep}\over 72 e^{1/2} \ln s}. 
 \]

In the SIMO case 
the probability bound (\ref{180}) becomes
\beq
\label{180'}
\lt(\lt(1 -c_1\sqrt{n-1\over m}\rt)^{n(n-1)}-2\delta\rt) \lt(1-2\tau-s^{-q}\rt),\quad n\geq 2. 
\eeq

\end{theorem}

\commentout{\begin{remark}
Typically  the total number of grid points $m$ is much
larger than the target sparsity $s$ and the dimension $np$ of
the measurement vector. In this case the probability bound (\ref{180}) 
is close to one. 

\end{remark}
}
\begin{remark}
Our results can be extended to
the case that the  sampling angles are not independent of
the incident angles by adjusting the probability 
(\ref{23-3}) (and hence (\ref{180}) and (\ref{180'})) in the  spectral norm 
bound. An important example is the multistatic data matrix with $\tilde\theta_j=-\theta_j, j=1,...,n$. Such a  setting is employed in the  well known and widely used MIMO imaging scheme called  MUSIC (standing for
 {\rm MU}ltiple-{\rm SI}gnal-{\rm C}lassification)  with
 $p=n$ \cite{Che,The}
 (see the Conclusion). 
 
Previously in \cite{cs-par}, we have applied
the compressed sensing methodology  to  imaging
with the multi-static data matrix under the paraxial approximation and the assumption that
 the point scatterers lie on a transverse plane. 
 
\end{remark}
\begin{proof}
First, the coherence estimate (\ref{mut}) and the sparsity constraint (\ref{spark}) implies (\ref{M}). 

Now  the norm bound (\ref{norm}) 
implies (\ref{Op})  if 
\beq
 \label{Op2}
 3\lt({q\ln{s}\over 2\ln{{m\over \tau}}}\rt)^{1/2}+{2s\over np}\leq {1\over 4 e^{1/4}},\quad q>1. 
 \eeq
 Hence for $np\gg s$ we can choose $q$ in (\ref{Op})  to be 
 \[ q={\ln{m} -\ln{ \tau} \over  72\sqrt{e}\ln s}. 
  \]

Since Theorems \ref{thm4} and \ref{thm5-2} hold simultaneously 
with probability
greater than 
\[
\lt(1-c_1\sqrt{np-1\over m}\rt)^{n(n-1)p(p-1)}-2\delta,\quad p, n\geq 2
\]
and since the target ensemble is independent
of the sensor ensemble we have the
 bound (\ref{180})  for the probability of
exact  recovery. The proof for the SIMO case
 is the same. 

This completes the proof of Theorem \ref{thm1}. 

\end{proof}

\subsection{SIMO/MISO exact inverse scattering}
\label{simo-miso}

Next we turn to the exact inverse scattering (\ref{10.11})-(\ref{fl}) which takes into account all  the multiple scattering
effects. As the previous simulation shows  \cite{cs-par},
multiple scattering can severely degrade  the 
performance of the imaging method based on
the Born approximation. 

Consider both the SIMO measurement with $p=1$
and  the MISO measurement with $n=1$. Note that
in the MISO measurement, various plane waves
are incident upon the scatterers  {\em one at a time}
and the corresponding  scattering amplitudes are sampled at
a fixed direction. 
By the reciprocity of wave propagation in a time-invariant
medium, reversing  the incident and scattered  waves
and interchanging their roles leave the scattering amplitude
unchanged, 
\[
A(\hat\br, \bd)=A(-\bd,-\hat\br)
\]
 (see Appendix \ref{recipro} for a proof). 
Therefore, the SIMO and MISO cases are equivalent
to each other. 
 
Hereafter we will  restrict our attention to the SIMO case.
To this end, we will work with  the alternative definition of the  target vector $X=(\nu_j u(\br_j))\in \IC^m$.
The reason for this is that the problem then has
the {\em appearance}  of  linear system
\beq
\label{10-10}
Y=\bPhi X
\eeq
where the sensing matrix $\bPhi$ has the  entries  
\beq
\label{10-11}
\Phi_{lj}=e^{-i\om(z_j\sin{\tilde\theta_l}+x_j\cos{\tilde\theta_l})}
\eeq
and is {\em independent} of the incident field. 
We then  apply  Theorem \ref{thm1} with $p=1$ and sufficiently
large $n$  to
recover $X$ with high probability. To recover $\nu$ from
$X$, we observe that 
 as long as $u(\br_{i_j})\neq 0, \forall j$, the support of $\nu$ is the same as 
 that of $X$. 
Indeed, 
the target strength  $\nu=(\nu_j)\in \IC^m$
can be recovered exactly by solving the system of nonlinear
equations on the target support as follows.

Define the illumination and full field vectors
at the locations of the scatterers:
\beqn
 U^{\rm i}&=&(u^{\rm i}(\br_{i_1}),...,u^{\rm i}(\br_{i_s}))^T\in \IC^s\\
U&=&(u(\br_{i_1}),...,u(\br_{i_s}))^T\in \IC^s.
\eeqn

Let $\bG$ be the $s\times s$ matrix 
\[
\bG=[(1-\delta_{jl})G(\br_{i_j},\br_{i_l})]
\]
 and $\cV$ the diagonal matrix 
 \[
 \cV={\rm diag}(\nu_{i_1},...,\nu_{i_s}). 
 \]
The Foldy-Lax equation  (\ref{fl}) can be written as
\beq
\label{fl2}
U=U^{\rm i}+\om^2 \bG\cV U
\eeq
from which  we obtain 
\beq
U&=&
\lt(\bI-\om^2\bG\cV\rt)^{-1}
U^{\rm i} \label{ill}
\eeq
and 
\beq
\label{ls3}
X=\cV U=\cV\lt(\bI-\om^2 \bG\cV\rt)^{-1}U^{\rm i}
\eeq 
provided that $\om^{-2}$ is not an eigenvalue of 
$\bG\cV$.
\commentout{
By Gershgorin's circle theorem the matrix $\bI-\bG\cV$
is invertible if
\beq
\label{weak}
 \max_{l} \sum_{j\neq l}|G(\br_{i_l},\br_{i_j})\nu_{i_j}|<1.
\eeq 
}
The exciting field then  determines
the scattering amplitude by (\ref{10.11})
which yields the (nonlinear) system
\beq
\label{27}
Y=\bPhi X=\bPhi \cV \lt(\bI-\om^2\bG \cV\rt)^{-1} U^{\rm i}
\eeq
where the sensing matrix entries are given by (\ref{10-11}). 

\commentout{
and $\cV$ can be solved from
this  system of  nonlinear equation by
the formula
\beq
\label{exact}
\cV=\mbox{\rm diag} \lt[{X\over \bG X+U^{\rm i}}\rt]
\eeq
where the division  is carried out entry-wise 
as stated in the following.
}

 \begin{proposition}
 \label{prop4}
Suppose
 \beq
 \label{281}
  \om^{-2} \,\,\hbox{ is not an eigenvalue of
 the matrix}\,\, \bG\cV
 \eeq
 and   
  \beq
 \label{282'}
U^{\rm i} \quad\hbox{is not orthogonal to any row vector  of }
\lt(\bI-\om^2\bG\cV\rt)^{-1}.
\eeq
 \commentout{
 \beq
 \label{282}
U^{\rm i} \quad\hbox{is not orthogonal to any column of }
\lt(\bI-\om^2\bG\cV\rt)^{-1}
\eeq
}
 Then the solution $\cV$
 of (\ref{ls3})  is given by 
 \beq
\label{exact}
\cV=\mbox{\rm diag} \lt[{ X\over \om^2 \bG X+U^{\rm i}}\rt]
\eeq
where the division  is in the  entry-wise sense  (Hadamard product).
  In this case, $\mbox{\rm supp} (\cV)=\mbox{\rm supp} ( X)$.    \end{proposition}
 \begin{proof}
 Note that
 \[
 \cV \lt(\bI- \om^2\bG \cV\rt)^{-1}
 =\lt(\bI- \om^2\cV\bG\rt)^{-1}\cV.
 \]
 Hence  eq. (\ref{ls3}) can be written as
 \[
 X=\lt(\bI- \om^2\cV\bG\rt)^{-1}\cV U^{\rm i}
 \]
 or equivalently 
 \beq
 \label{76}
\lt(\bI-\om^2 \cV\bG\rt)  X=\cV U^{\rm i}.
 \eeq
   Solving (\ref{76}) for  the diagonal matrix $ \cV$ entry-by-entry, we obtain
(\ref{exact}) which is well-defined if
 \beq
 \label{283}
 \om^2\bG X+U^{\rm i} \quad\hbox{contains no zero component.}
 \eeq 

Since
\[
\om^2\bG X+U^{\rm i}= \lt(\bI-\om^2\bG\cV\rt)^{-1} U^{\rm i}=U
\]
 (\ref{283}) follows from (\ref{282'}). 
 \commentout{condition  (\ref{282}) implies 
 \beqn
\lt(\bI+ \bG\cV \lt(\bI-\bG\cV\rt)^{-1}\rt) {\mathbf e}_j\cdot U^{\rm i}
&=& \lt(\bI+ \bG\cV \lt(\bI-\bG\cV\rt)^{-1}\rt) U^{\rm i}\cdot {\mathbf e}_j\neq 0,\quad j=1,...,s, 
\eeqn
}

\end{proof}
\begin{corollary}
\label{cor:zero}
Condition  (\ref{283}) holds and hence (\ref{exact}) is well-defined if
\beq
\label{38'}
\om^2\|\bG\cV\|<1/2
\eeq 
where  $\|\cdot\|$ 
equals the maximum of the absolute row sums 
of the matrix corresponding to the operator norm
on $L^\infty$.

\end{corollary}
\begin{proof}

Clearly $U^{\rm i}+\om^2\bG X$ contains no zero entry if
\beq
\label{41}
\om^2\|\bG X\|<1
\eeq
since every component of $U^{\rm i}$ has modulus one. 
As 
 \[
 \|\bG X\|=\|\bG\cV \lt(\bI-\om^2\bG\cV\rt)^{-1}\|
 \]
(\ref{41})  follows from (\ref{38'}). 
 \end{proof}

\begin{theorem}
\label{thm:bi} Suppose $p=1$, (\ref{m-2}),  (\ref{281}) and  (\ref{282'}) hold.  
Let $X$ be  a BP solution for the system (\ref{10-10}) with the matrix entries (\ref{10-11}) according to Theorem \ref{thm1}. 
The formula
(\ref{exact})  recovers exactly the target of sparsity
\beq
\label{spark2}
s<\lt(8\ln{m\over \ep} \rt)^{-1}
\lt(\chi^{\rm s}+{\sqrt{2}K\over \sqrt{n}}\rt)^{-2}
\eeq
with probability at least as in  (\ref{180'})
  and  $\chi^{\rm s}$ satisfies the bound (\ref{21-4})
or (\ref{21-2}) depending on whether $\hbox{\rm supp}(\pdfs)$
contains a Blind Spot or not. 

\commentout{
On the other hand, if the target sparsity satisfies
\beq
s\leq {1\over 2} +{1\over 2} \lt(\chi^{\rm s}+{\sqrt{2}K\over \sqrt{n}}\rt)^{-1}
\eeq
then the target can be recovered exactly  by BP with probability greater than $(1-\delta)^2$.
}

\end{theorem}
   \begin{remark} \label{rmk:zero}

The resonance frequency violating (\ref{281}) 
is related to the transmission eigenvalue
for continuous  media where an analogous
 non-resonance
condition  is also needed to ensure
the existence and uniqueness of the solution to the inverse
scattering problem \cite{CCM, Nac}. 

If
$1$ is an eigenvalue of $\om^2\bG\cV$, 
the existence of solution for (\ref{fl2}) 
requires that  $U^{\rm i}$ be orthogonal to
the eigenspace of $\om^2\bG\cV$ corresponding to $1$.
Then other physical constraints  (such as the minimum energy solution) need to be taken 
into account in order to obtain a unique solution. 
   
The simplest example  for resonance is this: Two point scatterers  
have the strengths  $\nu_{i_1}, \nu_{i_2}$ such that
$\hbox{\rm sign}{(\nu_{i_1})}=\hbox{\rm sign} {(\nu_{i_2})}=\hbox{\rm sign}{(G^*(\br_{i_1}, \br_{i_2}))}$. Then 
\[
\bG\cV=\lt[\begin{matrix}
0& |\nu_{i_1}G(\br_{i_1}, \br_{i_2})|\\
|\nu_{i_2}G(\br_{i_1}, \br_{i_2})|&0
\end{matrix}\rt]
\]
is real and symmetric and  has the positive  eigenvalue $ \sqrt{|\nu_{i_1}\nu_{i_2}|}
|G(\br_{i_1}, \br_{i_2})|$. The resonance frequency is
$ {|\nu_{i_1}\nu_{i_2}|^{-1/4}}
|G(\br_{i_1}, \br_{i_2})|^{-1/2}$ in this case. 


   \end{remark}
   
   \begin{remark}
  In view of (\ref{ill}),  (\ref{282'}) means that
  $U$ has no zero component. In other words, 
  the target is not shadowed by itself  in any way. 
  
Since the negation of  (\ref{282'}) is an algebraic constraint, 
 (\ref{282'}) are  satisfied almost surely  
in the target ensemble under (\ref{281}).  
\end{remark}

\subsection{Stability w.r.t.  errors}
Here we  consider the situation  
 where the measurement  or model errors are present.
 In the former case, the data may be contaminated by noise. 
In the latter case, the errors may be due to, for instance, the 
fact that the targets are slightly off the grid. In this case the target vector $X$ does not represent the true targets exactly due to
 model mismatch and is only the best approximation given
 the model.  In either case,  the data vector can be written
 as 
 \beq
\label{8}
Y=\bPhi X+E
\eeq
where $E$ represents the errors. In the case of measurement noise, $E$ is independent
of the targets while in the case of model mismatch, $E$ depends
explicitly on the targets. 

Since $A(\hat\br,\bd)$ is an  analytic function
of both $\hat\br$ and $\bd$  and hence
for the given  noisy data, in general no solution
exists to the inverse scattering problem. Even if
a solution does exist, it does not depend continuously
on the measured data in any reasonable norm. 

To deal with the problem of ill-posedness 
we consider, instead of the Tikhonov regularization,  the $L^1$-regularization
\beq
\label{relax}
\min_{ Z} {1\over 2} \|Y-\bPhi Z\|_2^2+\lambda\|Z\|_1
\eeq
where $\lambda$ is the Lagrangian multiplier to be given below. 
This  is the Lagrangian form of the basis pursuit denoising (BPDN)
 \cite{BDE, CDS} and the  {\em lasso} in the statistics literature 
\cite{Tib, EHJT}. 
Let $\hat X$ be the  minimizer of (\ref{relax}). 

The starting point of our analysis 
is  the following result,  due to Tropp \cite{Tro3}, 
concerning the error bound and
the recoverability of the target support in the presence of noise. 
\begin{proposition}\cite{Tro3}
\label{prop2} Assume that $\bPhi$ has all unit columns and
 that  $\|E\|_2\leq\ep $. 

Suppose $\mu(\bPhi) s\leq 1/3$.
Then the minimizer $\hat X$ of (\ref{relax}) with $\lambda=2\ep$ is unique and its support is contained in  ${\rm supp}(X)$.
Moreover, 
\beq
\label{548}
\|\hat X-X\|_\infty\leq \lt(3+\sqrt{3/ 2}\rt)\ep . 
\eeq
\end{proposition}
\begin{remark}
\label{rmk07}
When the matrix $\bPhi$ has all unit entries, as in (\ref{10-11}), 
and the error term satisfies $\|E\|_2\leq n^{1/2} \ep$, 
then (\ref{548})  holds with $\lambda=2n\ep$.
\end{remark}

Define the reconstruction of $\cV$ to be 
 \beq
 \label{70}
 \hat \cV =\mbox{\rm diag}\lt[{\hat X\over U^{\rm i}+\om^2 \bG \hat X}\rt]
 \eeq
 in analogy to (\ref{exact}). 

Using Proposition \ref{prop2}  we 
derive an error bound and  a sufficient condition under which
the support of the reconstruction is exactly the same
as the original.

\commentout{
We state two results that we shall use.

\begin{theorem} \cite{DE,GN}
\label{thm:spark}
Given $Y\in \IC^N$, BP produces 
 the sparsest possible solution if
\[
\|X\|_0\leq  {1\over 2} \lt({1\over \mu(\bPhi)}+1\rt). 
\]
\end{theorem}
}

\commentout{
Analogous to Proposition \ref{prop4} and Corollary \ref{cor:zero}  we have the following. 
 \begin{proposition}
 \label{prop4'} Suppose
 \beq
 \label{exact2}
 U^{\rm i}+\bG \hat X \quad
 \hbox{ contains no zero entry (cf. (\ref{283})) }
 \eeq
 and define \beq
 \label{70}
 \hat \cV =\mbox{\rm diag}\lt[{\hat X\over U^{\rm i}+\bG \hat X}\rt]. 
 \eeq
 Then $\mbox{\rm supp} (\hat\cV)=\mbox{\rm supp} ( \hat X)$ and $\hat\cV$ satisfies 
 \beq
 \label{69}
\hat X= \hat\cV{ \lt(\bI- \bG \hat \cV\rt)^{-1}U^{\rm i}}.
 \eeq

 \end{proposition}
}
\commentout{
\begin{theorem}\label{thm5}
Assume $\|\bG\cV\|<1/2, \mu s\leq 1/3$ and
 \beq
 \label{71-2}
b_0\equiv {1-2\|\bG\cV\|\over 1-\|\bG\cV\|}
>
 (3+\sqrt{3/2})\ep\|\bG \|
 \eeq
 and define the reconstruction 
 by
 \beq
 \label{70}
 \hat \cV =\mbox{\rm diag}\lt[{\hat X\over U^{\rm i}+\bG \hat X}\rt]. 
 \eeq
Then (\ref{70})  is well-defined and 
satisfies the error bound:
\beq
\label{90}
\|\cV-\hat\cV\|&\leq&
{\lt(1+ 2\|\bG\| \rt) (3+\sqrt{3/2})\ep \over 
b_0(b_0-(3+\sqrt{3/2})\ep\|\bG\|)}.
\eeq
If, in addition, 
\beq
b_0> (3+\sqrt{3/2})\ep \|\cV^{-1}\|, 
\eeq
then $\mbox{\rm supp}(\hat 
\cV)=\mbox{\rm supp}(\cV)$. 
\end{theorem}
}

\begin{theorem}\label{thm5}
Suppose $\|E\|_2 \leq \ep n^{1/2}$ and  let
$\hat X$ be the solution to (\ref{relax}) with $\lambda=2\ep n$. 
Assume 
\beq
\label{0.3}
\mu(\bPhi) s\leq 1/3
\eeq
 and 
\beq
 \label{71-2}
\om^2\|\bG\cV\|< {1- (3+\sqrt{3/2})\ep\|\bG \|\over
2- (3+\sqrt{3/2})\ep\|\bG \|}\,\,\lt( <{1\over 2}\rt). 
 \eeq
Then (\ref{70})  is well-defined and 
satisfies the error bound:
\beq
\label{90}
\|\cV-\hat\cV\|\leq
 {2\lt(1 +{\om^2\|\bG\|\|\cV\|}\rt) (3+\sqrt{3/2})\ep \over b_0( b_0-\om^2 (3+\sqrt{3/2})\ep\|\bG\|)},\quad b_0\equiv {1-2\om^2\|\bG\cV\|\over 1-\om^2\|\bG\cV\|}. 
\eeq
Moreover, $\hbox{\rm supp} (\hat \cV)=\hbox{\rm supp}(\hat X)\subset \hbox{\rm supp}(X)$. 
On the other hand, if
\beq
 \label{71-5}
\om^2\|\bG\cV\|&<& {1- (3+\sqrt{3/2})\ep\|\cV^{-1} \|\over
2- (3+\sqrt{3/2})\ep\|\cV^{-1}\|}
\eeq
then $\hbox{\rm supp} (\hat X)=\hbox{\rm supp} (X)$.
Therefore under (\ref{71-2}) and (\ref{71-5}),
 $\mbox{\rm supp}(\hat 
\cV)=\mbox{\rm supp}(\cV)$, i.e. the support 
of the target is perfectly recovered. 
\end{theorem}
The proof of Theorem \ref{thm5} is given in Section \ref{sec:noise}. 

\begin{remark}
Condition (\ref{0.3}) is slightly stronger than (\ref{spark3})
and hence the OMP algorithm can be used to solve
(\ref{relax}) \cite{Tro3}. 

Condition (\ref{71-2}), (\ref{71-5}) and (\ref{38'}) all say 
in various ways  that the scatterers are
either weak or far apart. 
\end{remark}

\commentout{
If one is concerned only with recovering the support
of the target, then the conditions (\ref{0.3}), (\ref{71-2})
and (\ref{71-5}) can be relaxed by using
a result of Cand\'es and Plan  (Theorem 1.3 of \cite{CP})
which considers, instead of (\ref{relax}), a different
relaxation parameter 
\beq
\label{relax2}
\min_{ Z} {1\over 2} \|Y-\bPhi Z\|_2^2+{\ep\over \sqrt{n}} \cdot 2\sqrt{2\log{m}} \|Z\|_1.
\eeq
\begin{theorem} Assume that $E=(E_j) \in \IC^n$ and $E_j, j=1,...,n$ are
i.i.d. complex Gaussian  r.v.s with variance $\sigma^2$.
Suppose that
\beq
\label{1002}
\mu(\bPhi)\leq \Phi_0/\log{m}
\eeq
and
\beq
\label{1003}
s\leq {C_0 m\over \|\bPhi\|^2 \log{m}}.
\eeq
Assume also that  (\ref{38'}) 
and
\beq
\label{1001}
{1-2\om^2\|\bG\cV\|\over 1-\om^2 \|\bG\cV\|}|\nu_j|
> {8\ep\over \sqrt{n}} \cdot \sqrt{2\log{m}},\quad\forall j
\eeq
hold.  Then,   with probability at least
$1-2m^{-1}((2\pi \log{m})^{-1/2}+s m^{-1}) -\cO(m^{-2\log2})$, $\hbox{supp}(\hat\cV)=\hbox{supp}(\cV)$ where
$\hat\cV$ is given by (\ref{70}) with $\hat X$ being
the solution of (\ref{relax2}).
\label{thm6}
\end{theorem}
\begin{remark}
Conditions (\ref{1002}) and (\ref{1003}) can be replaced by
explicit, sufficient conditions such as 
\beq
&&\lt(\chi^{\rm i}+\sqrt{2} K\rt)\cdot\lt(\chi^{\rm s}+\sqrt{2} K\rt)\leq A_0/\log{m}\label{1007}\\
&& s\leq {C_0n}/\log{m}\label{1008}
\eeq
following from (\ref{mut}) and (\ref{norm}), respectively. 
In some sense, these conditions tend to be much weaker than (\ref{0.3}), (\ref{71-2})
and (\ref{71-5}). On the other hand, (\ref{1007})
and (\ref{1008}) are only proved  to hold
probabilistically. 
\end{remark}
}

\commentout{
On the other hand, once the support of the target  is known
then the mean square error of the least-squares estimator can be
estimated as follows. 

\begin{corollary}
 Let $S$ be the support
of the target $\nu$. Let $\bPhi_S$ be the column submatrix of $\bPhi$ restricted to the set $S$. Let $\hat X_S$ be
the least-squares estimator on $S$, i.e. 
\[
\hat X_S= (\bPhi_S^* \bPhi_S)^{-1} \bPhi^*_S Y.
\]
Denote by $\hat X \in \IC^n$ the trivial extension of $\hat X_S$ by filling in zeros. Under the assumptions of Theorem \ref{thm6} we have
\[
\IE\|\hat X-X\|_2^2\leq {s\ep^2\over n (1-\delta_S)}
\]
where $\delta_S$ the RIC. 

\end{corollary}
}

\commentout{
\subsection{Extended targets} 
Some of the above estimates, especially the coherence
bound,  can be easily carried out for extended targets.
Here we demonstrate how this can be done by using 
the Haar wavelet basis. 

The 2-dimensional Haar wavelet is given by
\beq
\label{haar}
\psi(x,z)=\lt\{
\begin{matrix}
1,& x, z\in [0, 1/2]\\
-1, & x\in [0,1/2],\,\,  z\in [1/2, 1]\\
-1,&x\in [1/2, 1],\,\, z\in [0,1/2]\\
1,& x, z\in [1/2, 1]
\end{matrix}
\rt.
\eeq
which has the following Fourier transform
\[
\hat\psi(\xi,\zeta)={-1\over 2\pi \xi\zeta} \lt(1-e^{-i\xi/2}\rt)^2\lt(1-e^{-i\zeta/2}\rt)^2.
\]
The following set of functions 
 \beq
 \psi_{\bp,\bq}(\br)=2^{-(p_1+p_2)/2}\psi(2^{-\bp}\br-\bq),\quad \bp,\bq\in \IZ^2,\quad \br=(x,z)\in \IR^2
 \eeq
 where and below we use the following notation 
 \[
 2^{-\bp}\br=(2^{-p_1}x,2^{-p_2}z)
 \]
 forms an orthonormal basis in $L^2(\IR^2)$ \cite{Dau}. 
 This is the Haar wavelet basis by which we can write
 the square-integrable (assumed) function  $\nu(\br)\cdot 
 u(\br)$ as
 \[
 \nu(\br)\cdot u(\br) =\sum_{\bp,\bq\in \IZ^2} c_{\bp,\bq} \psi_{\bp,\bq}(\br)
 \]
 where the coefficients $
 c_{\bp,\bq}=\lan \psi_{\bp,\bq}, \nu u\ran
 $ are square-summable.  
  
 In the Haar wavelet basis, the data vector has
 the components
 \beq
\label{7}
Y_k={ 2\pi} \sum_{\bp,\bq\in \IZ^2} 2^{(p_1+p_2)/2}c_{\bp,\bq} e^{-i\om 2^\bp\hat\br_k\cdot \bq}\hat\psi(\om 2^\bp \hat\br_k),\quad k=1,...,n. 
\eeq
Suppose now
\[
E=2\pi \sum_{|\bp|_\infty>p_*} \sum_{|\bq|_\infty>m_\bp}
2^{(p_1+p_2)/2}c_{\bp,\bq} e^{-i\om 2^\bp\hat\br_k\cdot \bq}\hat\psi(\om 2^\bp \hat\br_k)
\]
is norm-bounded by $\ep$, i.e. $\|E\|_2\leq \ep$
for some integers $p_*, m_\bp$. 
For $ |\bq|_\infty\leq m_\bp, |\bp|_\infty\leq p_*$ we define the target vector  $X=(X_l)$
to be 
\[
X_l= c_{\bp,\bq} 
\]
where the index $l$ is related to $\bp=(p_1,p_2), \bq=(q_1, q_2) \in \IZ^2$ by 
 \beqn 
l&=&\sum_{j_1=-p_*}^{p_1-1}\sum_{j_2=-p_*}^{p_2-1}(2m_{\mb j}+1)^2+(q_1+m_{\bp})(2m_{\bp}+1)+(q_2+m_\bp+1),\quad {\mb j}=(j_1, j_2)\in \IZ^2. 
\eeqn
In other words, $X\in \IC^m$ with
\[
m=\sum_{j_1=-p_*}^{p_*}\sum_{j_2=-p_*}^{p_*}
(2m_{\mb j}+1)^2. 
\]
With  the sensing matrix elements  defined as
\beq
\Phi_{k,l}&\equiv & 2\pi 2^{(p_1+p_2)/2} e^{-i\om 2^\bp\hat\br_k\cdot \bq}\hat\psi(\om 2^\bp \hat\br_k)\nn\\
&=&2^{-(p_1+p_2)/2} e^{-i\om (2^{p_1}q_1\alpha_k+
2^{p_2}q_2\beta_k )}
{-1\over \om^2 \alpha_k\beta_k} \lt(1-e^{-i \om 2^{p_1-1}\alpha_k}\rt)^2
\lt(1-e^{-i \om 2^{p_2-1} \beta_k}\rt)^2
\eeq
where $\alpha_k=\cos\tilde\theta_k,\beta_k=\sin\tilde\theta_k$ 
we can write (\ref{7}) as (\ref{8}). 

Instead of (\ref{relax}) we consider 
\beq
\label{relax2}
\min_{Z\in \IC^m} \|Z\|_1,\quad \hbox{s.t.}\,\, \|Y-\bPhi Z\|_2\leq \delta,\quad \delta\geq \ep
\eeq
and invoke the following result  by Donoho, Elad and Temlyakov(Theorem 3.1, \cite{DET})
\begin{proposition}
\label{prop6}
Suppose the sparsity $s$  of the target vector $X$  satisfies
\[
s< {1\over 4} (1+{1\over \mu (\bPhi)}). 
\]
Then the minimizer $\hat X$ of (\ref{relax2}) obeys the
error bound
\[
\|\hat X-X\|_2\leq {\ep+\delta\over \sqrt{1-\mu(\bPhi)(4s-1)}}.
\]
\end{proposition}
Note that (\ref{relax}) is just the Lagrangian form of (\ref{relax2}) for an appropriate $\delta$. 
}

\commentout{
\beqn
k&=&\sum_{j_1=-p_*}^{p'_1-1}\sum_{j_2=-p_*}^{p'_2-1}(2n_{\bj}+1)^2+(q'_1+n_{\bp'})(2n_{\bp'}+1)+(q'_2+n_{\bp'}+1),\quad |\bq'|_\infty\leq n_{\bp'},
\quad |\bp'|_\infty\leq p_*, 
\eeqn
}

\commentout{

\subsection{Smooth  extended targets}\label{sec:ext}
In this section we consider  smooth extended targets
and represent them by interpolating point targets on
a regular grid. We then apply Theorems \ref{thm4}  and \ref{thm5} for the approximate reconstruction of
the smooth extended targets.

Suppose the target function $v(\br)=\nu(\br) u(\br)$ is a compactly supported smooth
 function. This is indeed the case when $\nu$ is smooth
 and has a compact support. We write 
\beq
\label{530}
v(\br)=\int \delta(\br-\br') v(\br') d\br'
\eeq
where $\delta(\br)$ is the Dirac-delta function
and  consider the filtered discretization of (\ref{530}) as follows.
Substituting 
\[
g_\eta(\br)={1\over {2\pi}\eta^2} e^{-{|\br|^2\over 2\eta^2}}. 
\]
 for the Dirac-delta function on the right
hand side of (\ref{530}) produces the filtered version $v_\eta$ of  $v$, i.e.
\beq
\label{531}
v_\eta(\br)=\int g_\eta(\br-\br') v(\br') d\br'.
\eeq
Clearly, $v_\eta$ tends to $v$ in the Schwartz space as $\eta\to 0$.
For our purpose, we can use $g_\eta(\br)=
\eta^{-2} g(|\br|/\eta)$ for   any Schwartz function $g$.

Next we discretize (\ref{531}) by replacing the integral by
the Riemann sum of step size $\ell < \eta$. We obtain
\beq
\label{533}
v_{\eta,\ell}(\br)=\ell^2\sum_{\bp\in I}
g_\eta(\br-\ell\bq ) v(\ell\bq),\quad I\subset \IZ^2
\eeq
which is as smooth as the interpolation element $g_\eta$ is. 
\commentout{ the difference in
the Sobolev norm $\|v_{\eta, \ell}-v\|_{k,1}$ can be
made arbitrarily small where
\[
 \|v\|_{k,1}=\max_{j\leq k} \lt\|{\partial^k\over \partial x^j\partial z^{k-j}} v \rt\|_1,\quad k\in \IN. 
 \]
 }

Since $v$ has a  compact support, $I$ is a finite set. 
For simplicity let  $I$ be the square sublattice 
\[
I=\{\bq=(q_1, q_2): q_1, q_2= 1,...,\sqrt{m}\}
\]
of total cardinality $m$. Let $j=(q_1-1)\sqrt{m}+ q_2$. 
Define the target vector $X=(X_j)\in \IC^m$ with $X_j=v(\ell\bq)$.
Now we write
the data vector $Y$ 
in the form (\ref{8}) with  the sensing matrix elements 
\beq
\label{535}
\Phi_{lj}&=&{1\over  \hat g_\eta(\om)}
\int_{\IR^2} g_\eta(\br'-\ell\bq) e^{-i\om \hat\br_l\cdot \br'} d\br',\quad j=(q_1-1)\sqrt{m}+ q_2\\
&=& e^{-i\om \ell\hat\br_l\cdot \bq}\nn
\eeq
and the error term $E$ due to the filtered discretization. 
For sufficiently small $\eta, \ell$ we may assume $\|E\|_2\leq \ep n^{1/2}$ for a given $\ep>0$ (see Remark \ref{rmk07} and below). 
We aim to reconstruct $X$ approximately  by the compressed-sensing methodology. 

The crucial  observation is that the sensing matrix (\ref{535}) is
identical to (\ref{10-11}) with $(x_j, z_j)=\ell\bq$ and
any {\em isotropic} filter function.
Therefore
Theorems \ref{thm4} and \ref{thm5}  hold verbatim  for
 the case of smooth extended targets formulated
above.

It should be noted that while Theorem \ref{thm5-2} holds
for extended targets, it is perhaps not very useful 
in view of Remark \ref{rmk2} since both
$m$ and the sparsity $s$ are $\cO(\ell^{-2})$. 
Indeed, the recovery criterion as formulated in Proposition \ref{tropp} is not suitable for treating extended targets which
by definition are excluded from the assumed target ensemble.

How small must $\eta$ and $\ell$ be in order to ensure
that $\|E\|_2\leq \ep n^{1/2}$? This  
can be answered roughly as follows. First, by the inequality $\|E\|_2\leq \|E\|_\infty \sqrt{n}$
it suffices to have $\|E\|_\infty \leq \ep $. Now
consider the transformation
$\cT$, defined by 
\[
 \cT v(\hat \br)={1\over \hat g_\eta(\om)} \int v(\br') e^{-i\om \hat \br \cdot \br'} d\br', 
 \]
 from  the space of  Schwartz functions on $\IR^2$ to
 the space of Schwartz functions on the circle, cf. (\ref{sa}). By definition
  \[
 E(\hat\br_l)= \cT v(\hat\br_l) -\cT v_{\eta,\ell}(\hat\br_l). 
 \]
 \commentout{
Since $\mu(\bPhi)=\cO(n^{-1/2})$ for sufficiently large $\om$ and
$s=\cO(\ell^{-2})$ for extended targets  as noted above we would like
to have $n^{-1}=\cO(\ell^4)$ in order to apply Theorem \ref{thm5}.  This gives rise to
the rough sufficient
condition $\|E\|_\infty
\leq \ep \ell^2$  for $\|E\|_2\leq \ep$. 

For sufficiently large $\om$
we have
\[
\|\cT (v-v_{\eta,\ell})\|_\infty\leq c_k
\om^{-k} \|v-v_{\eta,\ell}\|_{k,1}
\]
for any $k\in \IN$ and some $c_k>0$.
Consequently 
 \beq
 \label{536}
 \|v-v_{\eta,\ell}\|_{k,1} \leq c_k^{-1} \om^k \ell^2\ep 
 \eeq
  would be
 an alternative sufficient condition for $\|E\|_2\leq \ep$. 
 Considering, for example, $k=2$ in (\ref{536}) and 
 taking into account $\om \ell\gg 1$ we obtain the
 sufficient condition
  \beq
 \label{537}
 \|v-v_{\eta,\ell}\|_{2,1}\leq C\ep
 \eeq
 with the constant $C=\cO(\om \ell)$. In other words,
 the higher the frequency is the less  accurate 
 the discrete approximation needs to be. 
}

}
\subsection{Three dimensions}
\label{sec:3}
Here we consider the extension of the coherence bound,
Theorem \ref{thm4}, to the three dimensional setting. 
The main point here is to demonstrate the
decoherence effect due to the extra dimension. 
For simplicity, we will not consider the  improved
performance as a result of  avoiding
Blind Spots.

Instead of a square lattice, the computational domain
is a cubic lattice of spacing $\ell$. Each side of
the cubic lattice has $m^{1/3}$ grid points.
In three dimensions, the scattering amplitude has the
same expression (\ref{10.11}), except that
 the sampling direction $\hat\br=(\tilde\alpha, \tilde\beta, \tilde\gamma)$  are   parametrized by two polar angles
$\tilde\theta, \tilde\phi$ as 
\beq
\label{17}
\tilde\alpha=\cos{\tilde\theta}\cos{\tilde\phi},\quad \tilde\beta=\cos{\tilde\theta}\sin{\tilde\phi},
\quad \tilde\gamma=\sin{\tilde\theta}
\eeq

 \begin{theorem}\label{thm6}
 Let the sensing matrix $\bPhi$ be given according to the sensor ensemble. 
Suppose (\ref{m-2}) holds for some constants $\delta$ and $K$
and suppose $\pdfi, \pdfs \in C^1$.  
Then   the sensing matrix (\ref{14.3}) satisfies the coherence bound
\beq
\label{mut'}
\mu(\bPhi) < \lt(
\chi^{\rm i}+{\sqrt{2} K\over \sqrt{p}}\rt)
\lt(\chi^{\rm s}+{\sqrt{2}K\over \sqrt{n}}\rt)
\eeq
 with probability greater than $(1-\delta)^2$
 where in general $\chii$ (resp. $\chis$) satisfies the bound
 \beq
 \label{21-3'}
&\chi^{\rm i}\leq {c (1+\om \ell)^{-1}  \|\pdfi\|_{1,\infty}}&\lt(\hbox{resp.}  \quad \chi^{\rm s}\leq {c   ( 1+\om \ell)^{-1}\|\pdfi\|_{1,\infty}}\rt).  
 \eeq
\end{theorem}
Consequently, the asymptotic behavior (\ref{hf})
sets in faster in three dimensions than in two dimension in general. 

\subsection{Diffraction tomography: point sensors or sources }
\label{sec:near2}
\begin{figure}[t]
\begin{center}
\includegraphics[width=0.4\textwidth]{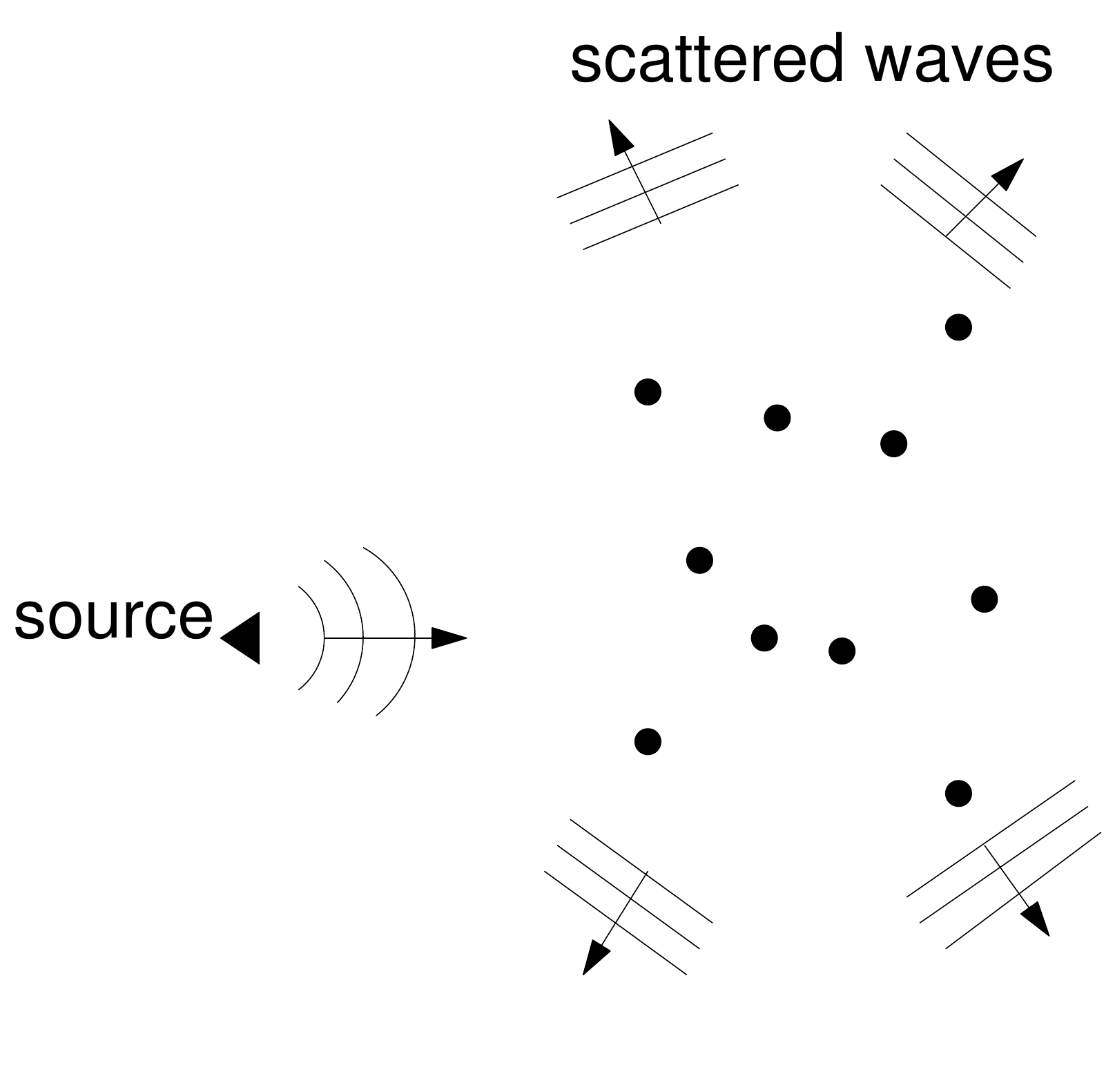}
\hspace{1cm}
\includegraphics[width=0.4\textwidth]{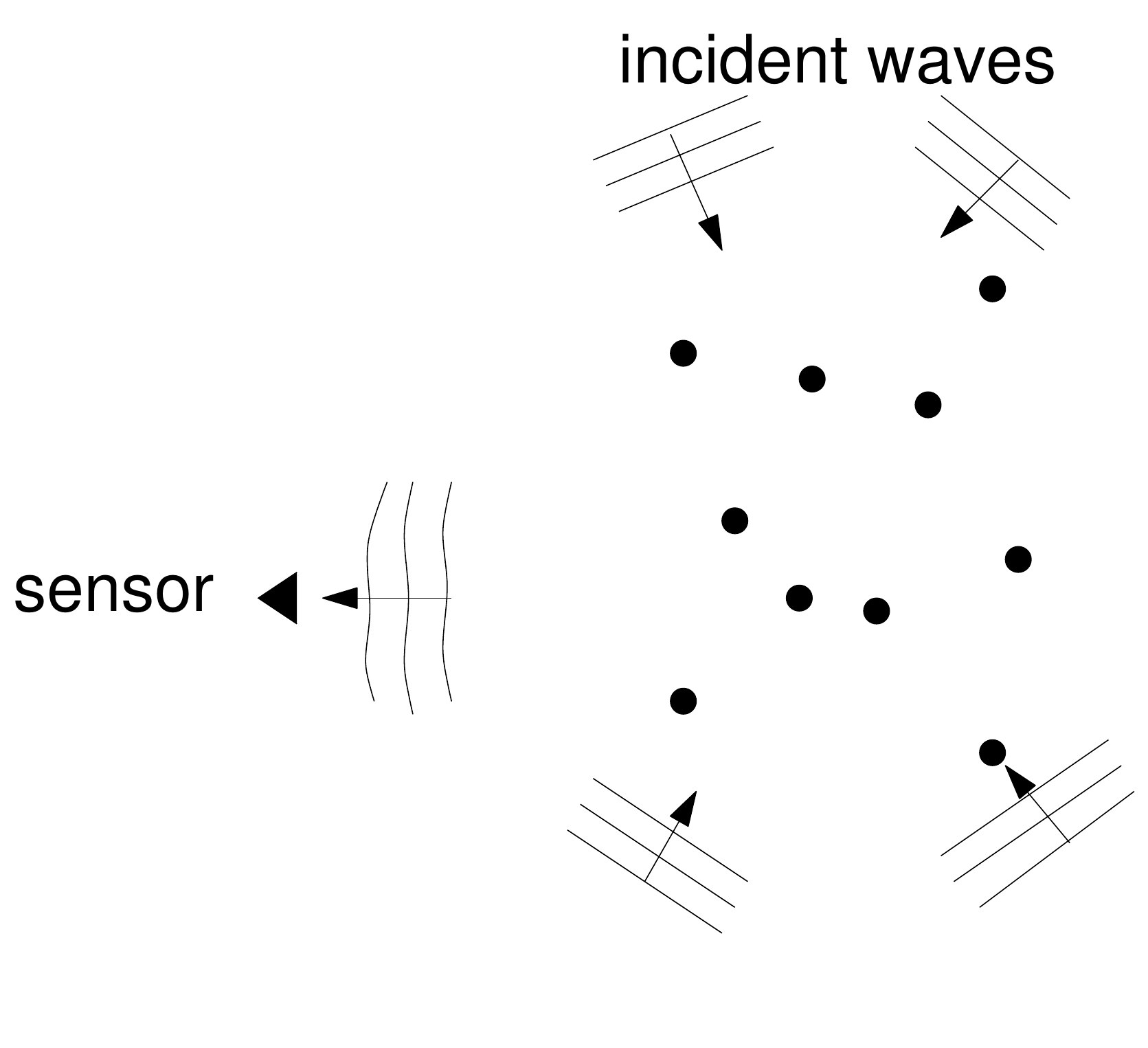}
\end{center}
\caption{SIMO (left) and MISO (right) measurements  with single point source (left) and single point sensor (right), respectively.}
\label{fig-dt'}
\end{figure}

\begin{figure}[t]
\begin{center}
\includegraphics[width=0.4\textwidth]{DT-geo.pdf}
\hspace{1cm}
\includegraphics[width=0.4\textwidth]{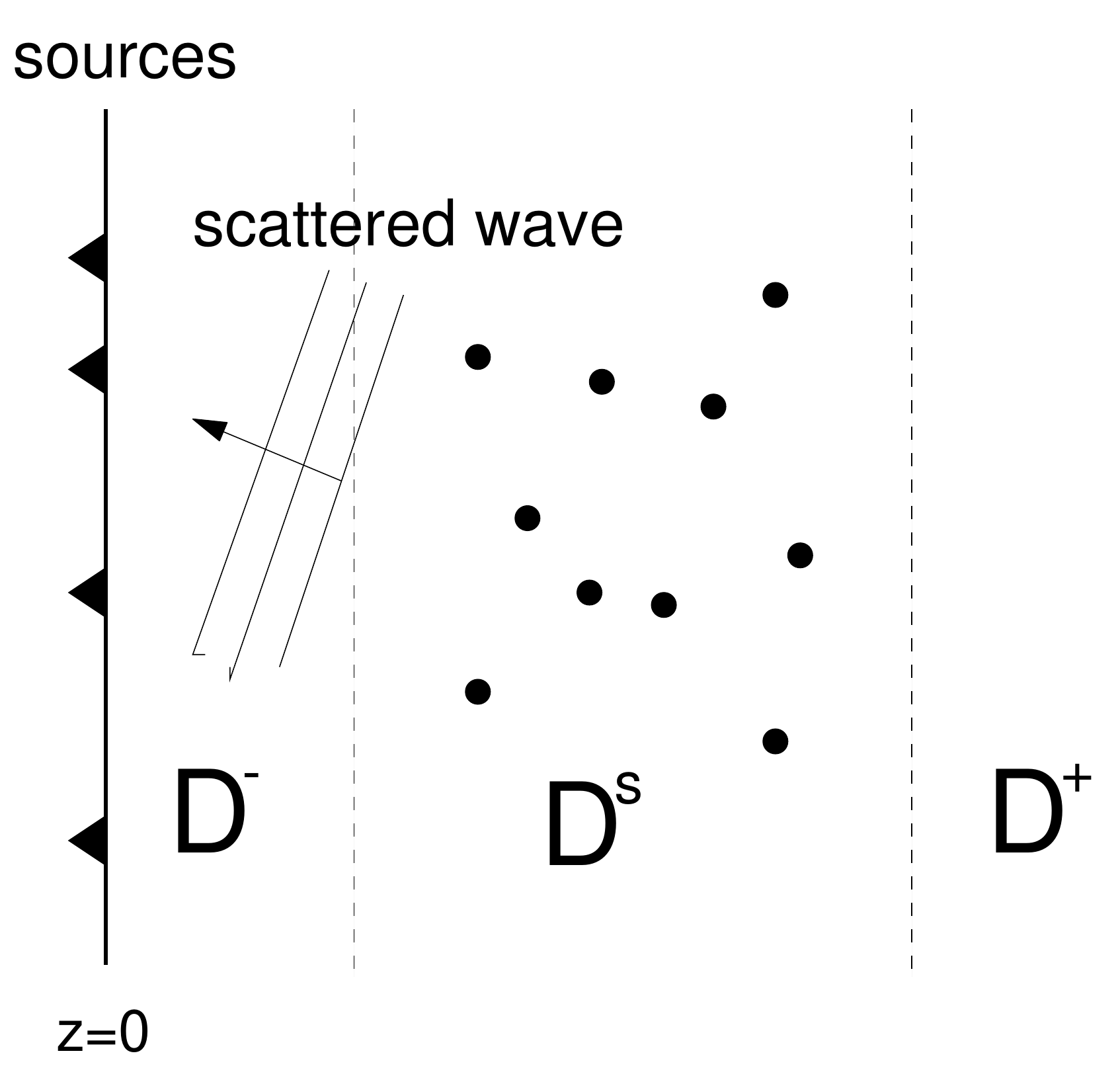}
\end{center}
\caption{SIMO (left) and MISO (right) measurements  with single incident plane wave (left) and single sampling direction  (right), respectively.}
\label{fig-dt}
\end{figure}

Instead of measuring the scattering amplitudes, one could 
measure  the scattered field at a set of  sampling points
and reconstruct  the targets from the measurement  data.
Likewise the incident wave may be a spherical wave emitted from a point source  instead
of a plane wave from far field. 
This is the measurement with point sources or sensors.   We will focus on the SIMO/MISO settings   and discuss both the two and
three dimensional cases.

First consider the simple setting with one point sensor at a fixed
location measuring the scattered fields due to  multiple incident plane waves (Figure \ref{fig-dt'}, right) emitted one at a time. This is
a  MISO measurement with one fixed point sensor.  
By the input-output  reciprocity (Appendix \ref{recipro}),
this is equivalent to the SIMO measurement of
 measuring multiple scattering amplitudes
due to one point source 
in near field (Figure \ref{fig-dt'}, left).  Hence it suffices to analyze the SIMO case.

Analogous to (\ref{10.11}), the scattering amplitude in the
direction $\hat\br$ is given by  
\beq
A(\hat\br, u^{\rm i})&=&{\om^2\over 4\pi}
\sum_{j=1}^m  \cG(\br_{j}, \br_0) \nu_j e^{-i\om\br_j\cdot\hat\br},\quad u^{\rm i}(\br)=G(\br,\br_0)\label{1058}
\eeq
where $\cG$ is the Green function including the multiple scattering effects, i.e.
\beq
\label{fltoo}
\cG(\br,\br_0)
&=&G(\br,\br_0)+\om^2 \sum_{j=1}^m \nu_{j} 
\cG(\br_{j},\br_0)G(\br, \br_{j}),\quad \br\neq \br_k, \,\,k=0,...,s\\
\cG(\br_k,\br_0)&=&G(\br_k,\br_0)+\om^2 \sum_{j\neq k} \nu_{j} 
\cG(\br_{j},\br_0)G(\br_k, \br_{j}),\,\,k=1,...,s\label{fltoo2}
\eeq
analogous to the Foldy-Lax equation (\ref{fl}). 
Define  the target vector $X=(X_j)\in \IC^m$ 
\[
X_j=\cG(\br_{j}, \br_0) \nu_j, 
\]
the data vector $Y=(Y_k)\in \IC^n$
\[
Y_k={4\pi \over \om^2} A(\hat\br_k)
\]
where $\hat\br_k=(\cos\tilde\theta_k, \sin\tilde\theta_k), k=1,...,n$ are the incident directions. 
Then we have $Y=\bPhi X$ where 
the sensing matrix $\bPhi=[\Phi_{kj}]$ is exactly as in (\ref{10-11}). 

We can solve this problem exactly as in Section \ref{simo-miso}
by first finding $X$ and then setting
\beqn
\nu_j&=&{X_j\over \cG(\br_j,\br_0)}\\
&=&{X_j\over G(\br_j,\br_0)+\om^2 \sum_{l=1}^m X_l  G(\br_0, \br_{l})}
\eeqn
where we have used (\ref{fltoo}). In the noisy case (\ref{8}), 
we proceed as before and set
\beqn
\nu_j
&=&{\hat X_j\over G(\br_j,\br_0)+\om^2 \sum_{l=1}^m \hat X_l  G(\br_0, \br_{l})}.
\eeqn
In other words, Theorems \ref{thm:bi} and \ref{thm5}
can be immediately generalized to this setting (one point source or one point sensor).

Next let us turn to  the more complicated measurement with
 multiple point sources and one sampling direction (MISO, Figure \ref{fig-dt}, right) 
or one incident wave and multiple sensors (SIMO, Figure \ref{fig-dt}, left). By the input-output  reciprocity (Appendix \ref{recipro}),
the two settings are  equivalent to each other. So we focus on
 the SIMO case below which requires similar but more
 delicate estimates than before. 

We are primarily interested in   the imaging set-up called {\em diffraction tomography} \cite{BW}. First, assume the two dimensional setting
with  the z-axis as  the imaging  direction (see Figure \ref{fig-dt}). 
To describe the near field measurement  we need  the Green function for  (\ref{helm}) 
  \[
  G(\br)={-i\over 4} H^{(1)}_0(\om|\br|)
  \]
  where $H^{(1)}_0$ is the zeroth order Hankel function of
  the first kind and admits 
  the Sommerfeld integral representation
  \beq
  \label{somm}
  H^{(1)}_0(\om |\br|)={1\over \pi}
  \int e^{i\om(|z|\gamma(\alpha)+x\alpha)} {d\alpha\over \gamma(\alpha)}
  \eeq
 with 
  \beqn
  \gamma(\alpha)=\lt\{\begin{array}{ll}
  \sqrt{1-\alpha^2}, & |\alpha|< 1\\
  i\sqrt{\alpha^2-1},& |\alpha|> 1
  \end{array}
  \rt.
  \eeqn
  \cite{BW}. 
  
  As shown in Figure \ref{fig-dt} the
whole space is divided  into  three domains $D^-\cup D^{\rm s}
\cup D^+$ where the infinite slab $D^{\rm s}$ contains
all the scatterers and $D^-, D^+$ are free half spaces.
The sensors are placed in either $D^-$ (the reflection mode)
or $D^+$ (the transmission mode).

\begin{figure}[t]
\begin{center}
\includegraphics[width=0.45\textwidth]{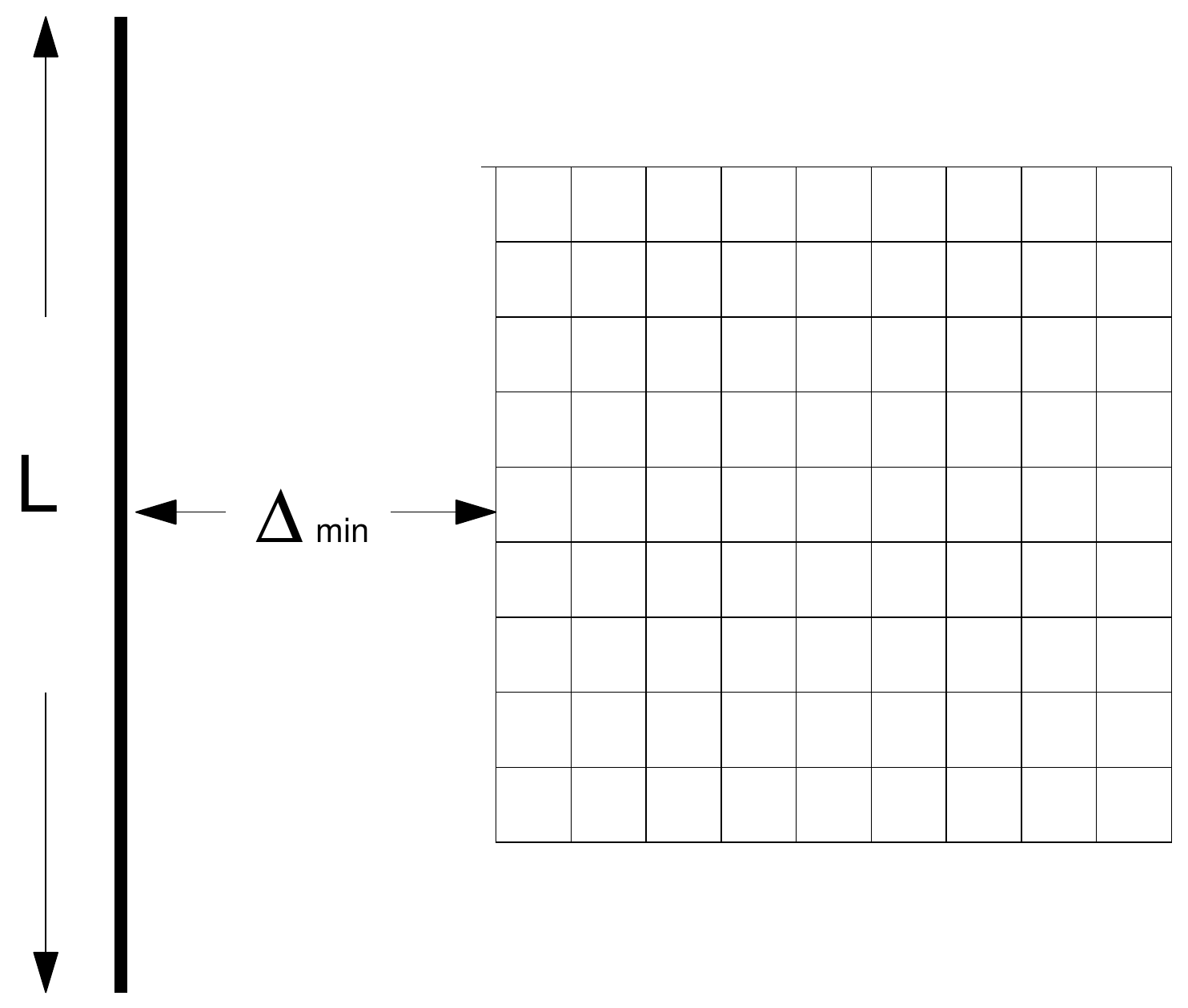}
\end{center}
\caption{Near-field imaging geometry}
\label{fig-near}
\end{figure}

To fixed the idea, suppose the sensors are located in 
the line segment of length $L$ on $\{z=0\} \subset D^-$ symmetrical w.r.t. to
the square lattice of the computational domain, see Figure
\ref{fig-near}. Let $\Delta_{\rm min}$ and $\Delta_{\rm max}
$, respectively,  be
the minimum and maximum distances  between the line segment 
and the lattice. In two dimensions,
\beq
\Delta_{\rm max}=\sqrt{{1\over 4} (L+\ell\sqrt{m})^2+
(\Delta_{\rm min}+\ell\sqrt{m})^2}.\label{1046}
\eeq

Using (\ref{somm}) we can express
the scattered field in $D^-\cup D^+$ 
as
\beq
\label{born}
u^{\rm s}(\br)
&=&{-i\om^2\over 4\pi}
\int_{D^{\rm s}} d\br' \int {d\alpha \over \beta}\nu (\br') u(\br')
e^{i\om(\beta|z-z'| +\alpha (x-x'))},\quad \br\in D^-\cup D^+.
\eeq

Let $\ba_j=(\xi_j, 0), j=1,...,n$ where $\xi_j$ are i.i.d.
uniform r.v.s  in a square  of length  $L$. Let 
$Y=(u^{\rm s}(\ba_j))\in \IC^n$ be  the data vector.
For the SIMO measurement  at near field 
 the sensing matrix elements, after factoring out $\om^2$,  are simply
\beqn
G(\ba_j-\br_l)&=&{-i\over 4\pi} \int {d\alpha \over \gamma}
e^{i\om(\gamma z_l +\alpha (\xi_j-x_l))},\quad j=1,...,n, \,\, l=1,...,m. 
\eeqn

\commentout{
In the special case of a  {\em linear} scatterer, 
then 
\beqn
\mu(\bPhi) \leq {2K^2\over \sqrt{pn}}+ {c_1\over ({\om \ell })^{3/2}}
\lt({\sqrt{2} K\over \sqrt{n}}+{c_2\ln{(\om L)}\over \om L}\rt)
\eeqn
for some constants $c_1, c_2>0$ with probability greater 
than $(1-\delta)^2$. 
}

For the three dimensional case, we use
the Weyl representation formula  
for  the Green function
 \beq
-{e^{i\om |\br|}\over 4\pi  |\br|}={-i\om \over 8\pi^2}
\int {d\alpha d\beta \over \gamma} 
 e^{\lt[i\om (\alpha x+\beta y+
 \gamma |z|)\rt]} \label{weyl}
 \eeq
 where
 \beqn
\gamma&=&\sqrt{1-\alpha^2-\beta^2},\quad \alpha^2+\beta^2\leq 1\\
\gamma&=&i\sqrt{\alpha^2+\beta^2-1},\quad \alpha^2+\beta^2>1,
\eeqn
\cite{BW}. 

 Using (\ref{weyl}) we can express
the scattered field  at $z=0$ as
\beqn
u^{\rm s}(\br)
&=&-{i\om^3\over 8\pi^2}
\int_{D^{\rm s}} d\br' \int {d\alpha d\beta \over \gamma}\nu (\br') u(\br')
e^{i\om(\gamma z' +\alpha (x-x')+\beta(y-y'))},\quad \br\in D^-\cup D^+.
\eeqn

Let $\ba_j=(\xi_j,\eta_j, 0), j=1,...,n$ where $(\xi_j,\eta_j)$ are i.i.d.
uniform r.v.s  in a square of length $L$. Assume again
that the square aperture is symmetrical w.r.t. 
the cubic lattice of the computational domain
as indicated in Figure \ref{fig-near}. For this imaging geometry,
the maximum distance $\Delta_{\rm max}$ between
the aperture and the lattice is 
\beq
\label{1046'}
\Delta_{\rm max}=
\sqrt{{1\over 4} (L+\ell m^{1/3})^2+
{1\over 4} (L+\ell m^{1/3})^2+(\Delta_{\rm min} +\ell m^{1/3})^2}. 
\eeq

 Let 
$Y=(u^{\rm s}(\ba_j))\in \IC^n$ be  the data vector.
The SIMO sensing matrix elements , after factoring out $\om^2$,  become
\beqn
G(\ba_j-\br_l) &=&{-i\om \over 8\pi^2} \int {d\alpha d\beta \over \gamma}
e^{i\om(\gamma z_l +\alpha (\xi_j-x_l)+\beta(\eta_j-y_l))}
\eeqn
for $j=1,...,n,  \,\, l=1,...,m$.

 \begin{theorem}\label{thm7}
Suppose  
\beq
\label{m}
m\leq {\delta\over 2} e^{2K^2/r_0^2},\quad \delta>0
\eeq
where $c_0$ depends on the minimum distance $\Delta_{\rm min} $ between
$\{z=0\}$ and the lattice (For $d=2$, $r_0=
\cO(-\log{\Delta_{\rm min}})$; for $d=3$, $r_0=\cO(\Delta_{\rm min}^{-1})$).

The mutual coherence obeys
\beq
\label{1043}
\mu(\bPhi)& \leq & |G(\Delta_{\rm max})|^{-2} \lt({\sqrt{2}K \over \sqrt{n}}+ {c\over \sqrt{\om L}}\rt),\quad d=2\\
\mu(\bPhi) &\leq  & |G(\Delta_{\rm max})|^{-2} \lt({\sqrt{2}K \over \sqrt{n}}+ {c\over {\om L}}\rt),\quad d=3\label{1044}
\eeq
for some constant $c$ (independent of $\om>0$ for $d=2$  and $\om>1$ for $d=3$),   with probability greater 
than $(1-\delta)^2$,  where $\Delta_{\rm max}$ is given by (\ref{1046}) for
$d=2$ and (\ref{1046'}) for $d=3$. 
\end{theorem}
\begin{remark}
In view of the presence of the factor 
$|G(\Delta_{\rm max})|^{-2}$, we see that
Theorem \ref{thm7} is useful primarily in  the high frequency
limit $\om \gg 1$ with $L=\cO(1)$.  
\end{remark}

Analogous to Theorems \ref{thm:bi} and \ref{thm5}, we have 
\begin{corollary}
\label{cor:bi-dt} 
In the absence of noise ($\ep=0$), 
let $X$ be  a BP solution with the sensing matrix of  diffraction tomography. Under the same assumptions of Theorem \ref{thm7} the formula
(\ref{exact})  recovers exactly the target of sparsity
smaller than 
\beq
{1\over 2} \lt(1+ 
 |G(\Delta_{\rm max})|^{2} \lt({\sqrt{2}K \over \sqrt{n}}+ {c\over \sqrt{\om L}}\rt)^{-1}\rt),&&d=2\\
{1\over 2}\lt(1+ |G(\Delta_{\rm max})|^{2} \lt({\sqrt{2}K \over \sqrt{n}}+ {c\over {\om L}}\rt)^{-1}\rt),&&d=3
\eeq
for some constant $c$ (independent of $\om>0$ for $d=2$  and $\om>1$ for $d=3$)
with probability greater than $(1-\delta)^2$. 

In the presence of noise,  Theorem \ref{thm5} holds. 
\end{corollary}

\section{Proof of Theorem \ref{thm4}: Coherence bound}
\label{sec:born}
\begin{proof}

 Denote $\hat\br_j
=(\cos\tilde\theta_j,\sin\tilde\theta_j),  \bd_k=(\cos\theta_k,\sin\theta_k)$.

The pairwise coherence has the form
\beq
\label{1.20-1}
{1\over p n}\lt|\sum_{k=1}^pe^{i\om\bd_k\cdot (\br-\br')}\cdot\sum_{j=1}^n e^{i\om\hat\br_j\cdot(\br-\br')}\rt|
\eeq
where $\br,\br'$ are two distinct points in the lattice $\cL$. 
Note that the two summations in (\ref{1.20-1}) are
of the same type.

  \commentout{
The Berry-Esseen theorem \cite{Fel}  states that the distribution of the sum of independent and 
identically distributed zero-mean random variables 
$\{X_i\}_{i=1}^n$
 normalized by its standard deviation, 
differ from the unit Gaussian distribution by at most 
$C \rho/(\alpha^2\sqrt{n})$, where $\alpha^2$ and $\rho$ are respectively the 
variance and the absolute third moment of the parent distribution, and $C$ is a distribution-independent 
absolute constant which is not greater than 0.7655
\cite{Sen}.
}

Consider the first summation over $k=1,...,p$. Let
\[
P_k=\cos{(\om\bd_k\cdot(\br-\br'))},\quad
Q_k=\sin{(\om\bd_k\cdot(\br-\br'))}
\]
and
\[
S_p=\sum_{k=1}^p P_k,\quad T_p=\sum_{k=1}^p Q_k.
\]
Then  the summation can be
bounded by
\beq
\label{57'}
\lt|\sum_{k=1}^pe^{i\om\bd_k\cdot (\br-\br')}\rt|\leq
\sqrt{|S_p-\IE S_p |^2+|T_p -\IE T_p|^2}+\sqrt{|\IE S_p|^2+|\IE T_p|^2}
\eeq
We recall the Hoeffding inequality
\cite{Hoe}.

\begin{proposition}
Let $P_1, ..., P_p$ be independent random variables. Assume
that $P_l \in [a_l, b_l], l=1,...,p$ almost surely.
Then we have 
\beq
\label{hoeff}
\IP\lt[\lt|S_p -\IE S_p\rt|\geq pt\rt]
\leq 2\exp{\lt[-{2p^2 t^2\over \sum_{l=1}^p (b_l-a_l)^2}\rt]}
\eeq
for all positive values of $t$. 
\end{proposition}

We apply the Hoeffding inequality to both $S_p$ and $T_p$. To
this end, we have $b_l-a_l=2, \forall l=1,...,p$.
Set
\[
t=K/\sqrt{p},\quad K>0. 
\]
Then we obtain 
\beq
\label{hoeff2}
\IP\lt[p^{-1}\lt|S_p -\IE S_p\rt|\geq K/\sqrt{p}\rt]
&\leq& 2e^{-{K^2/2}}\\
\IP\lt[p^{-1}\lt|T_p -\IE T_p\rt|\geq K/\sqrt{p}\rt]
&\leq& 2e^{-{K^2/2}}\label{hoeff22}.
\eeq

Note that the quantities $S_p, T_p$ depend on $
\br-\br' = (x_i-x_j, z_i-z_j)$ but they possess the symmetry:  $S_p(\br-\br')
=S_p(\br'-\br), T_p(\br-\br')=-T_p(\br'-\br)$. 
Furthermore, a moment of reflection reveals that thanks
to the square symmetry of the lattice 
there
are at most $m-1$ different values $|S_p|$ and $|T_p|$
among the $m(m-1)/2$  pairs  of $(\br, \br')$.

We use (\ref{hoeff2})-(\ref{hoeff22}) and the union bound to obtain
\beqn
&&{\IP\lt[\max_{i\neq j}p^{-1}\lt|S_p-\IE S_p\rt|\geq K/\sqrt{p}\rt]} \leq 2(m-1) \cdot e^{-K^2/2}\\
&&{\IP\lt[\max_{i\neq j}p^{-1}\lt|T_p-\IE T_p\rt|\geq K/\sqrt{p}\rt]}\leq 2(m-1) \cdot e^{-K^2/2}
\eeqn
where the factor $4 m$ is due to the structure of square lattice.

Hence, by (\ref{57'}) 
\beq
&&{\IP\lt[\max_{i\neq j}p^{-1}\lt|\sum_{k=1}^pe^{i\om\bd_k\cdot (\br-\br')}- \IE\lt[\sum_{k=1}^pe^{i\om\bd_k\cdot (\br-\br')}\rt] \rt|< \sqrt{2} K/\sqrt{p}\rt]} \label{10.1}>(1-2(m-1) e^{-K^2/2})^2.
\eeq
Similarly we have for the second summation in (\ref{1.20-1})
\beq
&&{\IP\lt[\max_{i\neq j}n^{-1}\lt|\sum_{j=1}^ne^{i\om\hat\br_j\cdot (\br-\br')}- \IE\lt[\sum_{j=1}^ne^{i\om\hat\br_j\cdot (\br-\br')}\rt] \rt|<\sqrt{2} K/\sqrt{n}\rt]} \label{10.2}>\lt(1-2(m-1)e^{-K^2/2}\rt)^2.
\eeq
By  (\ref{m-2})
 the right hand side of (\ref{10.1})-(\ref{10.2})
is greater than $(1-\delta)^2$. 

\commentout{
Consider  $ \IE\lt[\sum_{j=1}^ne^{i\om\hat\br_j\cdot (\br-\br')}\rt]$. The same analysis below applies equally well
to $\IE\lt[\sum_{k=1}^pe^{i\om\bd_k\cdot (\br-\br')}\rt] $.
If $\pdfs$ is the  uniform distribution over $[-\pi, \pi]$ or
$[-\pi/2, \pi/2]$ then 
\[
\IE\lt[\sum_{j=1}^ne^{i\om\hat\br_j\cdot (\br-\br')}\rt]  =nJ_0(\om |\br-\br'|)
\]
where $J_0$ is the zeroth order Bessel function. 
In general, the exact expression is not available but
we are concerned only with  the asymptotic for $\om|\br-\br'|\gg 1$.

The Bessel function has the asymptotic
 \beq
 \label{bess}
J_0(\om r)=\sqrt{2\over \pi\om r}\lt\{
\cos{(\om r-\pi/4)}+\cO((\om r)^{-1})\rt\},\quad \om r\gg 1. 
\eeq
That is,  for $\om \ell \gg 1$ there exists a constant
$c>0$ such that
\beq
\label{asym}
J_0(\om |\br-\br'|) < {c\over \sqrt{\om \ell}},\quad
\forall \br, \br'\in \cL,\quad\br\neq \br'
\eeq

In general, 
\beq
{1\over n}\IE\lt[\sum_{j=1}^ne^{i\om\hat\br_j \cdot (\br-\br')}\rt]
&=&\int_0^{2\pi} e^{i\om\tilde \bd \cdot (\br-\br')} \pdfs(\tilde \theta) d\tilde \theta, \quad \hat\bd=(\cos{\tilde\theta},
\sin{\tilde\theta})
\label{39-3}
\eeq
which is the Herglotz wave function with kernel
$\pdfs$ in two dimensions.
By assumption on $\pdfs$  (\ref{39-3}) is a finite sum of integrals of the form
\beq
\label{42-5}
\int_a^b e^{i\om\tilde \bd \cdot (\br-\br')} \pdfs(\tilde \theta) d\tilde \theta,\eeq
with  $\pdfs\neq 0$ in $(a,b)$ whose asymptotic for $\om\ell\gg 1$ can be
analyzed by the method
of stationary phase (Theorem XI. 14 and XI. 15 of \cite{RS}).

\begin{proposition} 
\label{prop:sph}
Let $g_{\br, \br'}(\tilde\theta)=\tilde\bd \cdot (\br-\br')/|\br-\br'|$ 
which is in $ C^\infty([-\pi,\pi]),\forall\br,\br'\in\cL$. 

(i) Suppose  ${d\over d\tilde \theta}g_{\br,\br'}(\tilde \theta)\neq 0, \forall
\tilde \theta \in [a,b],\forall \br,\br'\in\cL$.
Then for all $\pdfs\in C^h_0([a,b])$
\beq
\lt|\int e^{i\om|\br-\br'| g_{\br,\br'}(\tilde\theta)} \pdfs(\tilde\theta) d\tilde \theta\rt|
\leq c_h(1+\om|\br-\br'|)^{-h} \|\pdfs\|_{h,\infty}
\eeq
for some constant $c_h$ independent of $\pdfs$.
Moreover, since $\{g_{\br,\br'}:\br,\br'\in \cL\}$ 
is a compact subset of $C^{h+1}([a,b])$, the constant $c_h$
can be chosen uniformly for all $\br,\br'\in \cL$. 

(ii) Suppose ${d\over d\tilde \theta} g_{\br,\br'}(\tilde \theta)$  vanishes at $\theta_*\in (a,b)$. Since $ {d^2\over d\tilde \theta^2} g_{\br,\br'}(\theta_*)\neq 0$,  there exists a constant $c_t, t>1/2$
such that  
\beq
\lt|\int e^{i\om |\br-\br'| g_{\br,\br'}(\tilde \theta)}\pdfs(\tilde\theta)d\tilde \theta
\rt|\leq c_t (1+\om|\br-\br'|)^{-1/2}\|\pdfs\|_{t,\infty} 
\eeq
where the constant $c_t$ is independent of $ \br,\br'\in\cL$. 

 \end{proposition}
}
Consider $\IE\lt[\sum_{k=1}^pe^{i\om\bd_k\cdot (\br-\br')}\rt] $. The same analysis below applies equally well
to  $ \IE\lt[\sum_{j=1}^ne^{i\om\hat\br_j\cdot (\br-\br')}\rt]$.

If $\pdfi$ is the  uniform distribution over $[-\pi, \pi]$ or
$[-\pi/2, \pi/2]$ then 
\[
\IE\lt[\sum_{k=1}^pe^{i\om\bd_k\cdot (\br-\br')}\rt]  =pJ_0(\om |\br-\br'|)
\]
where $J_0$ is the zeroth order Bessel function. 
In general, the exact expression is not available but
we are concerned only with  the asymptotic for $\om|\br-\br'|\gg 1$.

The Bessel function has the asymptotic
 \beq
 \label{bess}
J_0(\om r)=\sqrt{2\over \pi\om r}\lt\{
\cos{(\om r-\pi/4)}+\cO((\om r)^{-1})\rt\},\quad \om r\gg 1. 
\eeq
That is,  for $\om \ell \gg 1$ there exists a constant
$c>0$ such that
\beq
\label{asym}
J_0(\om |\br-\br'|) < {c\over \sqrt{\om \ell}},\quad
\forall \br, \br'\in \cL,\quad\br\neq \br'
\eeq

In general, 
\beq
{1\over p}\IE\lt[\sum_{k=1}^pe^{i\om\bd_k\cdot (\br-\br')}\rt]
&=&\int_0^{2\pi} e^{i\om\bd\cdot (\br-\br')} \pdfi(\theta) d\theta,\quad \bd=(\cos\theta,\sin\theta) \label{39-3}
\eeq
which is the Herglotz wave function with kernel
$\pdfi$ in two dimensions.
By assumption on $\pdfi$  (\ref{39-3}) is a finite sum of integrals of the form
\beq
\label{42-5}
\int_a^b e^{i\om\bd\cdot (\br-\br')} \pdfi(\theta) d\theta
\eeq
with  $\pdfi\neq 0$ in $(a,b)$ whose asymptotic for $\om\ell\gg 1$ can be
analyzed by the method
of stationary phase (Theorem XI. 14 and XI. 15 of \cite{RS}).

\begin{proposition} 
\label{prop:sph}
Let $g_{\br, \br'}(\theta)=\bd\cdot (\br-\br')/|\br-\br'|$
which is in $ C^\infty([-\pi,\pi]),\forall\br,\br'\in\cL$. 

(i) Suppose  ${d\over d\theta}g_{\br,\br'}(\theta)\neq 0, \forall
\theta \in [a,b],\forall \br,\br'\in\cL$.
Then for all $\pdfi\in C^h_0([a,b])$
\beq
\lt|\int e^{i\om|\br-\br'| g_{\br,\br'}(\theta)} \pdfi(\theta) d\theta\rt|
\leq c_h(1+\om|\br-\br'|)^{-h} \|\pdfi\|_{h,\infty}
\eeq
for some constant $c_h$ independent of $\pdfi$.
Moreover, since $\{g_{\br,\br'}:\br,\br'\in \cL\}$ 
is a compact subset of $C^{h+1}([a,b])$, the constant $c_h$
can be chosen uniformly for all $\br,\br'\in \cL$. 

(ii) Suppose ${d\over d\theta} g_{\br,\br'}(\theta)$  vanishes at $\theta_*\in (a,b)$. Since $ {d^2\over d\theta^2} g_{\br,\br'}(\theta_*)\neq 0$,  there exists a constant $c_t, t>1/2$
such that  
\beq
\lt|\int e^{i\om |\br-\br'| g_{\br,\br'}(\theta)}\pdfi(\theta)d\theta
\rt|\leq c_t (1+\om|\br-\br'|)^{-1/2}\|\pdfi\|_{t,\infty} 
\eeq
where the constant $c_t$ is independent of $ \br,\br'\in\cL$. 

 \end{proposition}

\commentout{
Since the derivative of $[\bd\cdot (\br-\br')/|\br-\br'|]$   vanishes in the first order at a Blind Spot (\ref{bs}), 
if $\hbox{\rm supp}(\pdfi)$ contains the Blind Spot, then 
(\ref{42-5}) is $\cO((\om\ell)^{-1/2})$ according to
Theorem 13.1, Chapter 3, of \cite{Olv}. On the other hand, 
if $\hbox{\rm supp}(\pdfi)$ does not contains any  Blind Spot, then 
(\ref{42-5}) vanishes faster than any negative power of
$\om\ell$ by repeated application of Theorem 13.2, Chapter 3, of \cite{Olv}. Hence
\beqn
\chii&=&\lt|\int \lt[e^{i\om\bd\cdot (\br-\br')}\rt]\pdfi (\theta) d\theta\rt|
\eeqn
is $\cO((\om\ell)^{-1/2})$ if $\hbox{\rm supp}(\pdfi)$ contains a Blind Spot and is $o((\om\ell)^{-k}), \forall k\in \IN$ if $\hbox{\rm supp}(\pdfi)$ does not
 contains any  Blind Spot.
The similar result holds for
\beqn
\chis&=&\lt|\int\lt[e^{i\om\tilde \bd\cdot (\br-\br')}\rt]\pdfs d\theta\rt|. 
\eeqn
}
Note that the condition 
\[
{d\over d\theta}g_{\br,\br'}(\theta)\neq 0,\quad\theta\in  [a,b],\quad\forall \br,\br'\in\cL
\]
is the same as saying that $(a,b)$ does not contain
any Blind Spot.
Combining the estimates for 
\beqn
\chii&=\lt|\int \lt[e^{i\om\bd\cdot (\br-\br')}\rt]\pdfi (\theta) d\theta\rt|,\quad 
\chis&=\lt|\int\lt[e^{i\om\tilde \bd\cdot (\br-\br')}\rt]\pdfs d\tilde\theta\rt|, 
\eeqn
using (\ref{10.1})-(\ref{10.2})  and 
the identity
\[
UV=(U-\bar U)(V-\bar V)+\bar U(V-\bar V)
+\bar V(U-\bar U)+\bar U\bar V
\]
we 
 obtain (\ref{mut})
with probability greater than $(1-\delta)^2$. 

\end{proof}

\section{Proof of Theorem \ref{thm5-2}: Spectral norm bound}\label{sec:spec}
\begin{proof}
For the proof, it suffices to show that
the matrix $\bPhi$ satisfies
\beq
\label{gram}
\|{1\over m}\bPhi \bPhi^*-\bI_{np}\|_2<1
\eeq
where $\bI_{np}$ is the $np\times np$ identity matrix
with the corresponding probability bound (\ref{23-3}). 
By the Gershgorin circle theorem, (\ref{gram}) would in turn
follow from 
\beq
\label{19}
\mu \lt(\bPhi^*\rt)< {1\over np-1}
\eeq
since the diagonal elements of $\bPhi \bPhi^*/m$ are
unity.

The pairwise coherence amounts to calculating the expression 
\beq
\nn
{np\over m}\lt|\sum_{l=1}^m \Phi_{jl}\Phi^*_{li}\rt|
&=&{1\over m} \lt|\sum_{l=1}^me^{-i\om (z_l\sin{\theta}+x_l\cos\theta)} e^{i\om(z_l\sin{\theta'}+x_l\cos{\theta'})}
e^{i\om (z_l\sin\theta+x_l\cos\theta)}
e^{-i\om (z_l\sin\theta'+x_l\cos\theta')}\rt|
\eeq
Summing over $\br_l, l=1,...,m$ results in finite geometric series 
in the longitudinal and transverse coordinates since they are equally spaced. We obtain
\beq
\nn {np\over m}\lt|\sum_{l=1}^m \Phi_{jl}\Phi^*_{li}\rt|
&=&{1\over m} \lt| {e^{i\om\ell(\cos\theta'-\cos\theta+\cos\tilde\theta-\cos\tilde\theta')\sqrt{m}/2}-e^{-i\om(\cos\theta'-\cos\theta+\cos\tilde\theta-\cos\tilde\theta')\sqrt{m}/2}\over e^{i\om\ell (\cos\theta'-\cos\theta+\cos\tilde\theta-\cos\tilde\theta')} -1}\rt|\\
&& \times \lt|{e^{i\om\ell(\sin\theta'-\sin\theta+\sin\tilde\theta-\sin\tilde\theta)\sqrt{m}/2}-e^{-i\om(\sin\theta'-\sin\theta+\sin\tilde\theta-\sin\tilde\theta')\sqrt{m}/2}\over e^{i\om\ell (\sin\theta'-\sin\theta+\sin\tilde\theta-\sin\tilde\theta')} -1}\rt|.
\label{1.21}
\eeq
Using the identity $
|1-e^{i\phi}|=2|\sin{(\phi/2)}|
$
we then obtain 
\beq
\label{22}
{np\over m}\lt|\sum_{l=1}^m \Phi_{jl}\Phi^*_{li}\rt|
&=&{1\over m}  {\lt|\sin{\lt[\om\ell(\cos\theta'-\cos\theta+\cos\tilde\theta-\cos\tilde\theta')\sqrt{m}/2\rt]}\rt|\over \lt| \sin{\lt[\om\ell (\cos\theta'-\cos\theta+\cos\tilde\theta-\cos\tilde\theta')/2\rt]}\rt|}\\
&&\times 
  {\lt|\sin{\lt[\om\ell(\sin\theta'-\sin\theta+\sin\tilde\theta-\sin\tilde\theta')\sqrt{m}/2\rt]}\rt|\over \lt| \sin{\lt[\om\ell (\sin\theta'-\sin\theta+\sin\tilde\theta-\sin\tilde\theta')/2\rt]}\rt|}
\label{1.21'}\nn
\eeq
Since  $\theta,\theta'$  are independently and  identically
distributed according to $\pdfi$, the sine and cosine 
of these variables have the density functions 
\[
g_1^{\rm i}(t)={1\over \sqrt{1-t^2}}\pdfi(\arccos t)\quad
\hbox{and}\quad 
g_2^{\rm i}(t)={1\over \sqrt{1-t^2}}\pdfi(\arcsin t),
\]
respectively. 
Similarly the sine and cosine of $\tilde\theta,\tilde\theta'$
have the density functions
\[
g_1^{\rm s}(t)={1\over \sqrt{1-t^2}}\pdfs(\arccos t)\quad
\hbox{and}\quad 
g_2^{\rm s}(t)={1\over \sqrt{1-t^2}}\pdfs(\arcsin t), 
\]
respectively.

Hence  the random variables
\beqn
Z_1&=\om\ell (\cos\theta'-\cos\theta+\cos\tilde\theta-\cos\tilde\theta')/2&\in [-2\om\ell, 2\om\ell]\\
Z_2&=\om\ell (\sin\theta'-\sin\theta+\sin\tilde\theta-\sin\tilde\theta')/2&\in[-2\om\ell, 2\om\ell]
\eeqn
have the density function  
\beq
f_{Z_i}&=&{1\over \om\ell}(g_i^{\rm i}*g_i^{\rm i}*g_i^{\rm s}*g_i^{\rm s})({2z\over \om \ell}),\quad i=1,2.
\eeq

Since  $g*g*g*g$ is bounded in $[-4, 4]$, 
we have 
\[
\quad  \| f_{Z_i} \|_\infty \leq {c_0\over \om \ell},\quad \om\ell\gg 1,\quad i=1,2. 
\]
for some constant $c_0>0$. 

Define 
\beq
\label{view}
\zeta=\min_{\theta, \theta', \tilde\theta,\tilde\theta'}\min_{k\in \IZ}
\lt\{|Z_1-\pi k|, |Z_2-\pi k|\rt\}
\eeq
and note
\[
\sin\zeta>{2\zeta \over \pi},\quad \zeta\in (0, \pi/2).
\]
Hence the probability that $\{\zeta >b\}$ for small $b>0$ is larger than 
\[
(1-c_1b)^{n(n-1)p(p-1)} 
\] 
where the power $n(n-1)p(p-1)$ accounts for  the number of
different pairs of random variables involved in
(\ref{view}). 

By 
the choice
\[
b=\sqrt{np-1\over{ m}}
\]
we deduce that 
\[
\mu\lt(\bPhi^*\rt)< {1\over m b^2}={1\over np-1}
\]
with probability larger than 
\[
\lt(1-c_1\sqrt{np-1\over{ m}}\rt)^{n(n-1)p(p-1)} 
\]

In the SIMO case $p=1$, (\ref{view}) becomes
\beq
\label{view2}
\zeta=\min_{ \tilde\theta,\tilde\theta'}\min_{k\in \IZ}
\lt\{|Z_1-\pi k|, |Z_2-\pi k|\rt\}. 
\eeq
Hence the probability that $\{\zeta >b\}$ for small $b>0$ is larger than 
\[
(1-c_1b)^{n(n-1)} 
\] 
With 
\[
b=\sqrt{n-1\over{ m}}
\]
it follows  that 
\[
\mu\lt(\bPhi^*\rt)< {1\over m b^2}={1\over np-1}
\]
with probability larger than 
\[
\lt(1-c_1\sqrt{n-1\over{ m}}\rt)^{n(n-1)}.  
\]

\end{proof}

\section{Proof of Theorem \ref{thm5}: Stability}\label{sec:noise}

Next, we give an estimate for
the smallest  component of the exciting
field vector $U=(1-\om^2 \bG\cV)^{-1}U^{\rm i}$. 
\begin{proposition} 
If 
\beq
\om^2 \|\bG \cV\|<1/2 \label{322}
\eeq
then
\beq
\label{33}
\lt\|{1\over \lt(\bI-\om^2\bG \cV\rt)^{-1}U^{\rm i}}\rt\|_\infty
\leq { 1-\om^2\|\bG\cV\| \over
1-2\om^2\|\bG\cV\|}\equiv {1\over b_0}.
\eeq
Here  for any vector $V$, $V^{-1}$ denotes the vector whose entries are 
the reciprocal of those of $V$. \label{prop1}
\end{proposition}
\begin{proof}
We write
\[
\lt(\bI-\om^2\bG \cV\rt)^{-1}U^{\rm i} = U^{\rm i}\odot (1\oplus R)
\]
with  
\[
R=(U^{\rm i})^{-1}\odot \lt (\om^2\bG\cV U^{\rm i}+(\om^2\bG\cV)^2U^{\rm i}+
...\rt)
\]
which  converges under (\ref{322}). Here $\odot$ and $\oplus$ denote
the entrywise (Hadamard) product and sum, respectively, of
two vectors. 
Hence
\beq
\label{34}
\lt\|{1\over \lt(\bI-\om^2\bG \cV\rt)^{-1}U^{\rm i}}\rt\|_\infty
\leq \lt\|{1\over U^{\rm i}\odot (1\oplus R)}\rt\|_\infty\leq
{1\over 1-\|R\|_\infty}. 
\eeq
We also have
\beq
\label{35}
\|R\|_\infty&\leq 
\lt(\om^2\|\bG\cV\|
+\om^4\|\bG\cV\|^2+...\rt)
&={\om^2\|\bG\cV\|\over
1-\om^2\|\bG\cV\|}.
\eeq
Substituting (\ref{35}) into (\ref{34}) we obtain the claimed bound (\ref{33}). 
\end{proof}

From the Foldy-Lax equation (\ref{ill}) and (\ref{33})  we have the following lower bound on the exciting field vector $U$
\beq
\label{79}
\lt\|U^{-1} \rt\|_\infty
&\leq&  1/b_0
\eeq
and hence $b_0\leq |u(\br_{i_j})|,\forall j.$

\begin{corollary} \label{cor1} Suppose $\mu(\bPhi) s\leq 1/3$ and 
\beq
\label{64}
 {b_0}> \lt(3+\sqrt{3/ 2}\rt)\ep  \|\cV^{-1}\|. 
\eeq 
Then ${\rm supp}(\hat X)={\rm supp} (X)$.
\end{corollary}
\begin{proof}
This follows immediately from the fact
\[
\min_{j}|X_j|=\min_{j}  |\nu_{i_j}u(\br_{i_j})|\geq  {b_0 \over \|\cV^{-1}\|}> (3+\sqrt{3/2})\ep
\]
and Proposition \ref{prop2}. 
\end{proof}
  
 \begin{proposition}
 \label{prop5}
The vector $U^{\rm i}+\om^2\bG \hat X$ contains
 no zero entry if $\om^2 \|\bG \hat X\|_\infty <1$. In particular,
this is true for the minimizer $\hat X$ of  Proposition
 \ref{prop2} under the additional assumption 
 \beq
 \label{71}
 b_0> \om^2 (3+\sqrt{3/2})\ep\|\bG \|.
 \eeq
 In this case, $\mbox{\rm supp} (\hat \cV)=\mbox{\rm supp} (\hat X)$. 
 \end{proposition}
 \begin{proof}
The following calculation is straightforward 
 \beq
 \label{72'}
 \|\bG\hat X-\bG X\|_\infty
 &\leq \|\bG\|\|\hat X-X\|_\infty
 &\leq (3+\sqrt{3/2})\ep\|\bG\|.
 \eeq
 Moreover, since
 \[
 \bG X=\bG \cV (\bI-\om^2\bG\cV)^{-1}U^{\rm i}
 \]
 we have the estimate
 \beq
 \label{71'}
 \|\bG X\|_\infty \leq 
 {\|\bG\cV \| \over
 1-\om^2\|\bG\cV\|}.
 \eeq
 Hence 
 \beq
\nn
 \om^2\|\bG\hat X\|_\infty&\leq &\om^2 \|\bG X\|_\infty
 + \om^2 \|\bG\hat X-\bG X\|_\infty\\
 &\leq&  {\om^2\|\bG\cV\|\over 1-\om^2\|\bG\cV\|}
 +\om^2(3+\sqrt{3/2})\ep\|\bG \|\nn\\
 &=& {1-b_0} +\om^2 (3+\sqrt{3/2})\ep\|\bG \|<1
 \label{73'}
 \eeq
 under the additional condition 
  (\ref{71}). 
 \end{proof}
 
The proof of Theorem \ref{thm5} can now be completed
 as follows.

 \begin{proof}

 First of all, (\ref{71-2}) is equivalent to (\ref{71})
 and by Proposition \ref{prop5} the formula (\ref{70})
 is well-defined. 
 
Subtracting (\ref{exact}) from (\ref{70}) we can estimate
as follows:  
\beqn
\nn\|\cV-\hat\cV\|&=&
\lt\|{U^{\rm i}\odot (X-\hat X)+ \om^2 (X-\hat X)\odot \bG X
-\om^2 X\odot \bG(X-\hat X)\over ( U^{\rm i}+\om^2\bG\hat X)
( U^{\rm i}+\om^2\bG X)}\rt\|\\
&\leq& {\lt(1+ \om^2\|\bG X\|_\infty  \rt) \|X-\hat X\|_\infty +\om^2\|X\|_\infty \|\bG X-\bG\hat X\|_\infty\over 1-\om^2 \|\bG\hat X\|_\infty}\times 
\|U^{-1}\|_\infty  
\eeqn
where we have used the identity 
\beq
\label{u}
U= \om^2\bG X+U^{\rm i}. 
\eeq

By  (\ref{322}) and (\ref{71'}) we find that 
\beq
\om^2\|\bG X\|_\infty &<&1
\eeq
And, since $X=(\nu_j u(\br_j))$,
\[
\|X\|_\infty\leq \|\cV\| \|U\|_\infty. 
\]
 This, (\ref{79}),  (\ref{72'}) and Proposition \ref{prop2} 
lead to the bound 
\beq
\label{100}
\|\cV-\hat\cV\|\leq
 {(2 +\om^2\|\bG\|\|\cV\| \|U\|_\infty) (3+\sqrt{3/2})\ep \over b_0( b_0-\om^2 (3+\sqrt{3/2})\ep\|\bG\|)}.
\eeq
In view of  (\ref{u}) and (\ref{71'}) 
we have the following bound 
\beq
\label{u2} 
\|U\|_\infty &\leq 1+\om^2 \|\bG X\|_\infty
&\leq  {1\over 1-\om^2\|\bG \cV\|}<2.  
\eeq
The claimed result (\ref{90}) now follows from (\ref{100}) and (\ref{u2}). 

Since  (\ref{71-5}) is equivalent with 
 (\ref{64}) 
it follows from 
Corollary \ref{cor1}, Propositions \ref{prop5} and \ref{prop4}
that $\hbox{\rm supp}(\hat \cV)=\hbox{\rm supp} (\cV)$. 
\end{proof}

\commentout{
\section{Proof of Theorem \ref{thm6}}
The pair-wise scalar product of two columns
\[
{\sum_{k} \Phi^*_{ k,l'} \Phi_{k,l}\over \|\Phi_{l'}\|_2\|\Phi_{l}\|_2},
\]
after canceling some common factors, 
can be reduced to  the following  fraction with
the numerator 
\beq
\label{1101}&& -
\sum_{k}e^{i\om\lt((2^{p_1'}q_1' -2^{p_1}q_1 )\alpha_k+(2^{p_2'}q_2'-2^{p_2}q_2)\beta_k\rt)}
{\alpha^{-2}_k\beta^{-2}_k} \\
&&\cdot \lt(1-e^{-i \om 2^{p_1-1}\alpha_k}\rt)^2
\lt(1-e^{-i\om 2^{p_2-1}\beta_k}\rt)^2
 \lt(1-e^{i \om 2^{p'_1-1}\alpha_k}\rt)^2
\lt(1-e^{i\om 2^{p'_2-1}\beta_k}\rt)^2\nn
\eeq
and the  denominator 
\beq
\label{1102}
&&\sqrt{\sum_{k} \alpha^{-2}_k\beta^{-2}_k \lt|1-e^{i \om 2^{p_1-1}\alpha_k}\rt|^4
\lt|1-e^{i\om 2^{p_2-1}\beta_k}\rt|^4}\\
&& \cdot\sqrt{\sum_k\alpha^{-2}_k\beta^{-2}_k \lt|1-e^{i \om 2^{p'_1-1}\alpha_k}\rt|^4
\lt|1-e^{i\om 2^{p'_2-1}\beta_k}\rt|^4}.\nn
\eeq
We will follow the argument for the proof of Theorem \ref{thm4}
to prove concentration inequalities separately for
(\ref{1101}) and (\ref{1102}). 

The application of the Hoeffding
inequality,  Proposition \ref{hoeff},  requires us to estimate
the expectation of the sum and the range of
each summand. 

Let us start with (\ref{1102}). Denote the sum in (\ref{1102}) by $S_n$. We have for the expectation of
the sum
\beqn
\IE S_n&=&\IE\lt\{\sum_{k} \alpha^{-2}_k\beta^{-2}_k \lt|1-e^{i \om 2^{p_1-1}\alpha_k}\rt|^4
\lt|1-e^{i\om 2^{p_2-1}\beta_k}\rt|^4\rt\}\\
&=& n\int {d\tilde\theta \pdfs (\tilde\theta)
\over \cos^2\tilde\theta\sin^2\tilde\theta}
\lt|1-e^{i \om 2^{p_1-1}\cos\tilde\theta}\rt|^4
\lt|1-e^{i\om 2^{p_2-1}\sin\tilde\theta}\rt|^4 \nn
\eeqn
which is  bounded from below 
 by restricting the integration to
 the region 
 \[
 \lt|1-e^{i \om 2^{p_1-1}\cos\tilde\theta}\rt|\cdot
 \lt|1-e^{i\om 2^{p_2-1}\sin\tilde\theta}\rt|> c \]
 for some  $c>0$. Since $\pdfs$ is equivalent to
 the Lebesgue measure, 
 the resulting lower bound is $\cO(n)$ uniformly for all
 $\om 2^{p_1-1}, \,\,\om 2^{p_2-1}\geq 1$. 

Since $|\alpha_k|\in [\lamb, \sqrt{1-\lamb^2}]$, 
the $k$-th summand lies in 
\beqn
\lt[0, {256\over \lambda^2 (1-\lamb^2)}\rt].
\eeqn
Hence
\beqn
\IP\lt[\max_{|\bp|_\infty\leq p_*} n^{-1}\lt|S_n -\IE S_n\rt|<  \sqrt{2} K/\sqrt{n}\rt]
&>& \lt(1-2(2p_*+1)^2 e^{-2^{-15}{K^2\lamb^4(1-\lamb^2)^2}}\rt)^2.
\eeqn
In other words,  with high probability $\min_{|\bp|\leq p_*} S_n=\cO(n)$ for $K$ sufficiently large but much less than $\sqrt{n}$. 

Next consider (\ref{1101}) and denote the sum by $T_n$. 
Both the real and imaginary parts of each summand lie in the range 
\[
\lt[- {256\over \lambda^2 (1-\lamb^2)}, {256\over \lambda^2 (1-\lamb^2)}\rt]. 
\]
Now we need to estimate 
\beq
\label{1104}{1\over n} \IE T_n &=& \int {d\tilde\theta \pdfs (\tilde\theta)
\over \cos^2\tilde\theta\sin^2\tilde\theta}
e^{i\om\lt((2^{p_1'}q_1' -2^{p_1}q_1 )\cos\tilde\theta+(2^{p_2'}q_2'-2^{p_2}q_2)\sin\tilde\theta\rt)}
 \lt(1-e^{-i \om 2^{p_1-1}\cos\tilde\theta}\rt)^2\\
&&\cdot\lt(1-e^{-i\om 2^{p_2-1}\sin\tilde\theta}\rt)^2
 \lt(1-e^{i \om 2^{p'_1-1}\cos\tilde\theta}\rt)^2
\lt(1-e^{i\om 2^{p'_2-1}\sin\tilde\theta}\rt)^2\nn
\eeq
 using the method
of stationary phase.  
Rewriting  the fast phase in (\ref{1104}) as
\[
\om\lt((2^{p_1'}q_1' -2^{p_1}q_1 )\cos\tilde\theta+(2^{p_2'}q_2'-2^{p_2}q_2)\sin\tilde\theta\rt)
=\om \Delta_{\bp,\bq, \bp', \bq} g(\tilde\theta)
\]
with $\Delta_{\bp', \bq', \bp, \bq}=|2^{\bp'}\bq'-2^{\bp}\bq|_\infty$.
}

\section{Proof of Theorem \ref{thm6}}\label{sec:three}

Let $\bd=(\alpha,\beta,\gamma)$ be
parameterized by the angles $\theta,\phi$ as 
\beq
\label{17'}
\alpha=\cos{\theta}\cos{\phi},\quad \beta=\cos{\theta}\sin{\phi},
\quad \gamma=\sin{\theta}. 
\eeq
 The pairwise coherence has the form
\beq
\label{30}
{1\over p n}\sum_{k=1}^pe^{i\om(\alpha_k,\beta_k,\gamma_k)\cdot (\br-\br')}\sum_{j=1}^n e^{i\om(\tilde\alpha_j,\tilde\beta_j,\tilde\gamma_j)
\cdot(\br-\br')}
\eeq
where $(\alpha_k,\beta_k,\gamma_k), k=1,..., p$ and
$(\tilde\alpha_j,\tilde\beta_j,\tilde\gamma_j), j=1,..., n$ 
are independently and identically  distributed in the {\em unit} sphere according to $\pdfi(\theta,\phi)$ and $\pdfs(\theta,\phi)$, respectively. 

The main difference between  two and three dimensions
is in  evaluating the expectation of
$p^{-1}\sum_{k=1}^pe^{i\om(\alpha_k,\beta_k,\gamma_k)\cdot (\br-\br')}$ and $n^{-1}\sum_{j=1}^n e^{i\om(\tilde\alpha_j,\tilde\beta_j,\tilde\gamma_j)
\cdot(\br-\br')}$ which amounts to calculating
the integrals
\beq
\label{55-2}
\int^{\pi/2}_{-\pi/2}d\theta
\pdfii(\theta+\theta_0)\cos{\theta} \exp{\lt[i\om  |\br-\br'|\sin\theta\rt]}\\
\int^{\pi/2}_{-\pi/2}d\theta
\pdfsi(\theta+\theta_0)\cos{\theta} \exp{\lt[i\om  |\br-\br'|\sin\theta\rt]}\label{55-3}
\eeq
for some $\theta_0$ depending on $\br-\br'$
where $\pdfii$ and $\pdfsi$ are the marginal density functions
\beqn
\pdfii(\theta)&=&\int^\pi_{-\pi}d\phi \pdfi(\theta,\phi)\\
\pdfsi(\theta)&=&\int^\pi_{-\pi}d\phi \pdfs(\theta,\phi)
\eeqn

If  $\pdfii=\pdfsi=1/\pi$,
the integrals (\ref{55-2}) and (\ref{55-3})  become
\beqn
 {2\sin{\lt(\om |\br-\br'|\rt)}\over \om |\br-\br'|}
 =\cO({1\over \om\ell}),\quad \om \ell\gg 1. 
\eeqn
For the general case, integrating by parts with (\ref{55-2}) 
and (\ref{55-3}) 
produces 
\beq
{i\over \om |\br-\br'|}\lt[\pdfii(\theta+\theta_0)e^{i\om|\br-\br'|\sin\theta}\Big|^{\pi/2}_{-\pi/2}-\int^{\pi/2}_{-\pi/2}
e^{i\om|\br-\br'|\sin\theta}{d\over d\theta} \pdfii(\theta+\theta_0)d\theta\rt] \label{56-3}\\
{i\over \om |\br-\br'|}\lt[\pdfsi(\theta+\theta_0)e^{i\om|\br-\br'|\sin\theta}\Big|^{\pi/2}_{-\pi/2}-\int^{\pi/2}_{-\pi/2}
e^{i\om|\br-\br'|\sin\theta}{d\over d\theta} \pdfsi(\theta+\theta_0)d\theta\rt] \label{56-4}
\eeq
from which we obtain the bound
\beq
\label{57-1}
\lt|\int^{\pi/2}_{-\pi/2}d\theta
\pdfii(\theta+\theta_0)\cos{\theta} \exp{\lt[i\om  |\br-\br'|\sin\theta\rt]}\rt|\leq {c\over 1+\om \ell}\|\pdfi\|_{1,\infty}\\\
\lt|\int^{\pi/2}_{-\pi/2}d\theta
\pdfsi(\theta+\theta_0)\cos{\theta} \exp{\lt[i\om  |\br-\br'|\sin\theta\rt]}\rt| \leq {c\over 1+\om \ell}\|\pdfs\|_{1,\infty}. \label{57-2}
\eeq
\commentout{
If a Blind Spot is present in the sensor ensemble 
then there exists $\theta_0$ corresponding
to some $\br, \br'\in \cL$ such that 
\beqn
\label{57-3}
{d\over d\theta} \pdfii(\pm {\pi\over 2}+\theta_0)\neq 0.
\eeqn
As a consequence, (\ref{56-3}) is $\cO((\om\ell)^{-1})$. 

On the other hand, if no Blind Spot is present
in the sensor ensemble and
\[
{d\over d\theta} \pdfii(\theta+\theta_0)
=o\lt((\theta\pm \pi/2)^k\rt), \quad\forall k\in \IN
\]
for all $\br,\br'\in\cL$, then (\ref{56-3}) vanishes
faster than any negative power of $\om\ell\gg 1$
by repeated integration by parts or application
of Theorem 13.2, Chapter 3, \cite{Olv}. 
}

\section{Proof of Theorem \ref{thm7}}
\subsection{Two dimensional case}
The pairwise coherence is a ratio of the numerator
\beq
\label{1.20}
{1\over n}\lt|\sum_{j=1}^n G^*(\ba_j-\br)
G(\ba_j-\br')\rt| 
\eeq
and the denominator
\beq
\label{1.202}
{1\over n}\lt(\sum_{j=1}^n |G(\ba_j-\br)|^2\rt)^{1/2}
\cdot \lt(\sum_{j=1}^n |G(\ba_j-\br')|^2\rt)^{1/2}
\eeq
with $\br=(x,z),\br'=(x',z')$ being two distinct elements of $\{\br_l,l=1,...,m\}$. 
Note that $G^*(\ba_j-\br)
G(\ba_j-\br')$ are i.i.d. random variables since
$\ba_j$ are. 

To analyze the summation in (\ref{1.20})  we follow
the argument of the proof of Theorem \ref{thm4}. 
Define the random variables $P_j, Q_j, j=1,...,n$
to be the real and imaginary parts, respectively,
of $G^*(\ba_j-\br)
G(\ba_j-\br')$. We apply the Hoeffding inequality to both $S_n=\sum^n_{j=1} P_j$ and $T_n=\sum^n_{j=1} Q_j$. 

Note that the ranges $r_0=b_j-a_j$  of $P_j$ and $Q_j$  
depends on the minimum distance $\Delta_{\rm min}$ between
$\{z=0\}$ and the square lattice (Figure \ref{fig-near}).  Indeed for small distance $\Delta_{\rm min}\ll 1$,  $r_0=O(-\log{\Delta_{\rm min} })$. We have 
\beq
\label{1045} &&{\IP\lt[\max_{i\neq j}{1\over n}\lt|\sum_{j=1}^n
\lt[G^*(\ba_j-\br)
G(\ba_j-\br') -
\IE \lt(G^*(\ba_j-\br)
G(\ba_j-\br')\rt)\rt]\rt|<\sqrt{2}K/\sqrt{n}\rt]}\\
& &>\lt(1-\delta\rt)^2\nn
\eeq
if (\ref{m}) holds.

Using (\ref{somm}) we have the following calculation
for the expectation: 
\beq
\nn\lefteqn{\IE\lt\{ G^*(\ba_j-\br)
G(\ba_j-\br')\rt\}}\\
\nn &=& {1\over 16 \pi^2}
\int {d\alpha\over \gamma^*}{d\alpha'\over \gamma'}
\IE\lt\{e^{-i\om (z\gamma^*+(\xi_j-x)\alpha)}
e^{i\om(z' \gamma'+(\xi_j-x')\alpha')}\rt\}\\
&=& {1\over 16 \pi^2 }
\int {d\alpha\over \gamma^*}{d\alpha'\over \gamma'}
e^{i\om(-z\gamma^*+z'\gamma')}
e^{i\om(\alpha x -\alpha' x')}
{2\sin{\lt((\alpha'-\alpha)\om L/2\rt)}\over
\om L(\alpha'-\alpha)}.\label{1900}
\eeq
We divide the integration (\ref{1900}) into three parts:
$I_i, i=1,2,3$, the integrals over,
respectively,
\beqn
R_1&=&\{ |\alpha| >\sqrt{2}\}\cup\{ |\alpha'| >\sqrt{2}\}, \\
R_2&=& R_1^c
\cap\{ |\alpha-\alpha'|<H/(\om L)\},\\
R_3&=&R_1^c\cap \{ |\alpha-\alpha'|>H/(\om L)\}
\eeqn
  where $H\ll \om L$ will  be determined later.  

To estimate $I_1$, note that
\[
\lim_{\ep \to 0} {2\sin{\lt((\alpha'-\alpha)/\ep\rt)}\over
\alpha'-\alpha} = \delta(\alpha'-\alpha)\quad\mbox{in the sense
of distribution}. 
\]
Consequently $\om L \cdot |I_1|$ tends to
\beqn
 {1\over 16 \pi^2 }\lt|\int_{|\alpha|>\sqrt{2}}  {d\alpha\over |\gamma|^2}
e^{-\om |z+z'||\gamma|}
e^{i\om\alpha (x- x')}\rt|&\leq &
 {e^{-2\om \Delta_{\rm min}}\over 16 \pi^2 }\lt|\int_{|\alpha|>\sqrt{2}}  {d\alpha\over |\gamma|^2}\rt|,\quad \om L\to \infty, 
\eeqn
and hence is  bounded  uniformly in  $\om \Delta_{\rm min}>0$.  

For $I_2$ and $I_3$  we have the rough estimates
\beqn
|I_2|&\leq &{1\over 16 \pi^2}\int_{R_2} {d\alpha d\alpha'\over |\gamma\gamma'|}=\cO\lt({H\over \om L}\rt)\\
|I_3|&\leq &{1\over 16 \pi^2}\int_{R_3} {d\alpha d\alpha'\over |\gamma\gamma'|}{1\over H} =\cO\lt({1\over H}\rt)
\eeqn
where the singularities are inverse-square-root and hence  integrable. 
By choosing $H=\sqrt{\om L}$ we obtain 
\beq
\label{1042}
\lt|{\IE\lt\{ G^*(\ba_j-\br)
G(\ba_j-\br')\rt\}}\rt|
\leq {c_1\over \sqrt{\om L}},\quad \om L\gg 1
\eeq
for some constant $c_1>0$. 

A simple lower bound for the expression (\ref{1.202}) 
is given by
\beq
\label{1041}
\min_{\ba, \br} \lt|G(\ba-\br)\rt|^2
\eeq
where $\ba$ runs through the entire square aperture
of length $L$ and $\br$ runs through the entire square
domain of length $\ell\sqrt{m}$. Since $\lt|G\rt|$ is
a decreasing function of the absolute value of its argument,
(\ref{1041}) is achieved by the pair of points furthest away
from each other.    In the assumed 
imaging geometry, Figure \ref{fig-near},  the largest distance $\Delta_{\rm max}$  between any
pair of points is given by (\ref{1046}). 
Combining (\ref{1045}, (\ref{1042}) and (\ref{1041}) we obtain that (\ref{1043}) as claimed. 

\subsection{Three dimensional case}
We now need to estimate (\ref{1.20}) and (\ref{1.202})
with $\br=(x,y, z),\br'=(x',y', z')$ being two distinct elements of $\{\br_l,l=1,...,m\}$. We proceed as before. 

Define the random variables $P_j, Q_j, j=1,...,n$
to be the real and imaginary parts, respectively,
of $G^*(\ba_j-\br)
G(\ba_j-\br')$. 
In three dimensions,  the ranges $r_0$ of $P_j$ and $Q_j$  are 
$\cO(\Delta_{\rm min}^{-1})$. 

Using (\ref{weyl}) we have the following calculation:
\beqn
\lefteqn{{1\over \om^2} \IE\lt\{ G^*(\ba_j-\br)
G(\ba_j-\br')\rt\}}\\
\nn &=& {1\over 64 \pi^4}
\int {d\alpha d\beta \over \gamma^*}{d\alpha' d\beta' \over \gamma'}
\IE\lt\{e^{-i\om (z\gamma^*+(\xi_j-x)\alpha+(\eta_j-y)\beta)}
e^{i\om(z'\gamma'+(\xi_j-x')\alpha'+(\eta_j-y')\beta')}\rt\}\\
&=& {1\over 64 \pi^4 }
\int {d\alpha d\beta\over \gamma^*}{d\alpha' d\beta'\over \gamma'}
e^{i\om(-z\gamma^*+z'\gamma')}
e^{i\om(\alpha x-\alpha' x')}e^{i\om(\beta y-\beta' y')}\nn\\
&&\cdot {2\sin{\lt((\alpha'-\alpha)\om L/2\rt)}\over
\om L(\alpha'-\alpha)}\cdot 
{2\sin{\lt((\beta'-\beta)\om L/2\rt)}\over
\om L(\beta'-\beta)}. \nn
\eeqn
We divide the integration into several parts as follows.
Denote
\beqn
A_1&=&\{ |\alpha| >1\}\cup\{ |\alpha'|>1\},\\
A_2&=& \{|\alpha-\alpha'| >H/(\om L)\},\\
B_1&=& \{ |\beta|>1\}\cup
\{ |\beta'|>1\},\\
B_2&=&\{ |\beta-\beta'| >H/(\om L)\}. 
\eeqn
Let 
$I_i, i=1,2,3,5,6,7$, be the integrals over, respectively,  $
R_1=A_1\cap B_1, R_2=
A_1^c\cap B_1, R_3=A_1\cap B_1^c, R_4=A_1^c\cap B_1^c.$

$I_1$ can be estimated straightforwardly by 
\beqn
I_1&\leq &{c\over 64 \pi^4 \om^2 L^2 }
\int_{R_1}e^{-\om\gamma (z+z')} {d\alpha d\beta\over |\gamma|^2}
\eeqn
for a constant $c$ and is bounded uniformly in $\om >1$. 

$I_2$ can be further decomposed into two parts
parts $I_{2i}, i=1,2$ corresponding, respectively,  to integration over
$
R_{21}=A_1^c\cap B_1\cap A_2,
 R_{22}=
A_1^c\cap B_1\cap A^c_2 $
  and estimated by 
\beqn
I_2&=&I_{21}+I_{22}\leq {c\over 64 \pi^4 }\int {d\alpha d\beta d\alpha' d\beta'\over |\gamma \gamma'|}
\lt({1\over H}+{H\over \om L}\rt) \cdot{1\over \om L}
\eeqn
for a constant $c$. 

$I_3$ can be further decomposed into two 
parts $I_{3i}, i=1,2,$ corresponding, respectively,  to integration over
$
 R_{31}=A_1\cap B^c_1\cap B_2^c,
  R_{32}=A_1\cap B^c_1\cap B_2
 $
  and estimated by 
  \beqn
I_3&=&I_{31}+I_{32}\leq {c\over 64 \pi^4 }\int {d\alpha d\beta d\alpha' d\beta'\over |\gamma \gamma'|}
\lt({1\over H}+{H\over \om L}\rt) \cdot{1\over \om L}
\eeqn
for a constant $c$. 

$I_4$ can be decomposed into four parts $I_{4i}, i=1,2,3,4$ corresponding, respectively, to integration over
\beqn
R_{41}&=&A_1^c\cap A_2\cap  B^c_1\cap
B_2,\\
R_{42}&=& A^c_1\cap A_2\cap  B_1^c
\cap
B_2^c,\\
R_{43}&=& A_1^c\cap A^c_2\cap  B^c_1\cap
B_2,\\
R_{44}&=& A^c_1\cap A^c_2\cap  B_1^c
\cap
B_2^c,
\eeqn
and estimated by
\beqn
I_4&=&\sum_{i=1}^4I_{4i}\leq
{c\over 64 \pi^4 }\int {d\alpha d\beta d\alpha' d\beta'\over |\gamma \gamma'|}\lt({1\over H}\cdot{1\over H}
+{1\over H}\cdot {H\over \om L}+{H\over \om L}\cdot
{1\over H}+{H\over \om L}\cdot {H\over \om L}\rt)
\eeqn
for some constant $c$. 

\commentout{
\beqn
R_1&=&A_1\cap B_1,\\
R_2&=&\lt((A_1^c\cap B_1)\cup (A_1\cap B_1^c)\rt)
\cap\lt( \lt(A_2\cap
B_2^c \rt)\cup
\lt(A_2^c\cap
B_2\rt)\rt),\\
R_3&=& \lt((A_1^c\cap B_1)\cup (A_1\cap B_1^c)\rt)
\cap A^c_2\cap
B_2^c,\\
R_4&=&\lt((A_1^c\cap B_1)\cup (A_1\cap B_1^c)\rt)
\cap A_2\cap B_2,\\
 R_5&=&A_1^c\cap B_1^c \cap\lt( \lt(A_2\cap
B_2^c \rt)\cup
\lt(A_2^c\cap
B_2\rt)\rt),\\
R_6&=& A_1^c\cap B_1^c \cap A^c_2\cap
B_2^c,\\
R_7&=& A_1^c\cap B_1^c\cap A_2\cap B_2.
  \eeqn
  }

Setting $H=(\om L)^{1/2}$ we obtain
\beqn
{1\over \om^2} \IE\lt\{ G^*(\ba_j-\br)
G(\ba_j-\br')\rt\}
\leq { c_2\over \om L},\quad \om L \gg 1
\eeqn
for some constant $c_2>0$. 

Using (\ref{1041}) for the lower bound of (\ref{1.202})
we conclude (\ref{1044})  in Theorem \ref{thm7}.

\section{Conclusion}
We have analyzed the SIMO/MISO and MIMO inverse scattering  problems by 
compressed sensing theory to  shed new light on this  problem with distinguished history \cite{CCM, LP, Mel, RS}.  
We have obtained  several main results: Theorem \ref{thm1}
concerns the recoverability by $L^1$ minimization  with {the MIMO 
measurement  under the Born approximation, 
Theorem \ref{thm:bi} addresses the recoverability by 
$L^1$ minimization 
with the SIMO/MISO measurement including multiple scattering  and Theorem \ref{thm5}
asserts stability to measurement or model errors 
for weak or widely separated scatterers under 
a stronger sparsity constraint. We have also
analyzed the diffraction tomography with 
few views and limited angles and proved
the coherence bound (Theorem \ref{thm7}) analogous 
to Theorem \ref{thm1}. 

\commentout{
Three  issues   in the analysis merit 
further study:
 the notion of Blind Spots that may slow down the approach to
the high frequency asymptotic (Theorem \ref{thm1} and \ref{thm:bi}), the notion  of resonance  frequency 
and the notion of self-induced shadow which
can prevent the exact recovery  of the target support
 (Theorem \ref{thm:bi}). 
 }

A main limitation to our approach  is the assumption of  finite
(albeit high) 
dimensional targets. 
 Also, the reconstruction succeeds
only probabilistically. These limitations  are intrinsic to the current formulation
of the compressed sensing theory which is still evolving and in this regard  the present paper is only a first  step in the new direction. 

On the other hand, the compressed sensing  approach  is constructive and treats
the uniqueness and the reconstruction in a unified way.  
Indeed,  the main advantage of  this approach is an explicit and efficient
method (i.e. the Basis Pursuit with the SIMO/MIMO or MIMO  sensing matrix)  for reconstructing the scatterers 
from the scattering amplitude. Moverover,  the aperture can be
rather arbitrary   and   the dimension 
of  measurement  can be as low  as comparable to 
the target sparsity (up to a $\log{(m)}$-factor). 

It may be worthwhile to compare our imaging method
with  
 the MUSIC algorithm 
  which employs multiple sensors to collect 
 the $n\times n$ multistatic response data matrix
 where $n$ is the number of transmitters/receivers  \cite{Che, The}. 
 When the measurement is carried out in the far field, the $(l,j)$-entry of the response  matrix 
 is the measured scattering amplitude 
 for the sampling direction $l$  and the incident direction $j$. 
 It is not known if MUSIC can recover the target support exactly 
 for nonlinear inverse scattering.  Only the case for the Born approximation  has been shown capable of exact recovery of
 the target {\em support} in the absence of noise \cite{Kir}
(see
the corrected argument in Theorem 4.1, \cite{KG}). 
And the estimate for the required dimension of the measurement
for the  exact  recovery is hardly optimal. 
 This result should be compared to Theorem \ref{thm1} with
 $p=n$ and $\theta_j=-\tilde\theta_j, j=1,...,n$, in particular
 the sparsity constraint (\ref{spark}) for compressed sensing versus the necessary condition $n>s$  for MUSIC.   This 
 represents a significant reduction in the number of sensors
 when the  sparsity of the target vector is large. 

In a separate paper \cite{siso-cs} we propose novel  multi-shot  single-input-single-output (SISO) 
compressive  imaging methods and demonstrate
their superior performances including 
the capability of imaging {\em extended}  targets. 
We also present in \cite{siso-cs} numerical comparative  study  of the respective performances
of the  SIMO/MIMO  and multi-shot SISO schemes. 


\begin{appendix}
\section{Input-output reciprocity}\label{recipro}

More generally, consider the Helmholtz equation with a source
\[
\Delta u(\br)+\om^2 (1+\nu(\br)) u(\br)= - f(\br)
\]
which can be solved by using the Green function $\cG$ as
\[
u(\br)=\int \cG(\br,\br') f(\br') d\br'.
\]
By slight abuse of
notation,  we shall write the solution as 
\[
u=\cG f
\]
where $\cG$ stands also for   the corresponding propagator. 

Because the incident wave is governed by
\[
\Delta u^{\rm i}(\br)+\om^2 u^{\rm i}(\br)= - f(\br)
\]
we can write
\beqn
u(\br) &=&-\cG \lt[\Delta u^{\rm i}+\om^2 u^{\rm i}\rt](\br)\\
&=& \cG \lt[-(\Delta+\om^2(1+\nu)) u^{\rm i}+\om^2 \nu u^{\rm i}\rt](\br)\\
&=& u^{\rm i}(\br)+ \om^2 \cG\lt[ \nu u^{\rm i}\rt](\br).
\eeqn
Hence  the scattered wave $u^{\rm s}=u-u^{\rm i}$ can be
expressed as
\beq
\label{exact''}
u^{\rm s}(\br)=\om^2\int \cG(\br,\br') \nu(\br')  u^{\rm i}(\br')d\br'
\eeq
which is the  reciprocal representation  to (\ref{exact'}). 

For point scatterers, (\ref{exact''}) becomes
\beq
\label{dual}
u^{\rm s}(\br)=\om^2\sum_{j=1}^m \nu_{j}\cG(\br,\br_{j}) 
u^{\rm i}(\br_{j}),\quad \br\neq \br_k,\,\,k=1,...,s.
\eeq
Substituting the Foldy-Lax equation 
\beq
\cG(\br,\br_j)=G(\br,\br_{j})+\om^2\sum_{k\neq j} \nu_{k} \cG(\br_{j},\br_{k}) G(\br,\br_{k}),\quad \br\neq \br_j,\,\,j=1,...,s
\label{fl3}
\eeq
in (\ref{dual}) we obtain
\beq
u^{\rm s}(\br)&=&\om^2\sum_{j=1}^m \nu_{j} u^{\rm i}(\br_{j}) G(\br,\br_{j}) +\om^4\sum_{j=1}^m\sum_{k\neq j}
\nu_{j} \nu_{k} \cG(\br_{j},\br_{k}) u^{\rm i}(\br_{j}) G(\br,\br_{j})\nn
\eeq
which can be rewritten as
\beq
u^{\rm s}(\br)&=&\sum_{j=1}^m \sum_{k=1}^m  \delta_{j,k} \om^2\nu_{k} u^{\rm i}(\br_{k}) G(\br,\br_{j}) +\om^4\sum_{j=1}^m\sum_{k\neq j}
\nu_{j} \nu_{k} \cG(\br_{j},\br_{k}) u^{\rm i}(\br_{k})G(\br,\br_{j})\nn\\
&=&\sum_{j=1}^m \sum_{k=1}^m \lt[ \delta_{j,k} \om^2\nu_{k} +(1-\delta_{j,k})\om^4\nu_{j} \nu_{k} \cG(\br_{j},\br_{k})\rt] u^{\rm i}(\br_{k})G(\br,\br_{j}) \label{dual2}
\eeq
(see \cite{MG, Mar3} for a similar, but slightly  erroneous, expression). 

To obtain the alternative expression for the scattering amplitude,
let $\br \to \infty$ in (\ref{dual2}) and extract the plane wave spectrum by (\ref{somm}) for $d=2$ or (\ref{weyl}) for $d=3$.
The scattering amplitude in the direction $\hat \br$ is given by
\beq
A(\hat \br, u^{\rm i})={1\over 4\pi}\sum_{j=1}^m \sum_{k=1}^m \lt[ \delta_{j,k} \om^2\nu_{k} +(1-\delta_{j,k})\om^4\nu_{j} \nu_{k} \cG(\br_{j},\br_{k})\rt] u^{\rm i}(\br_{k})e^{-i\om \hat\br\cdot \br_{j}} \label{sa-dual}
\eeq
where $u^{\rm i}$  is not necessarily  a plane wave. 

In the case of a plane wave incidence (\ref{inc}) we observe
the symmetry between the incident and scattered plane waves
in (\ref{sa-dual}). Therefore,  reversing and interchanging roles
of  the incident and scattered waves do not affect the scattering amplitude, i.e. $A(\hat\br,\bd)=A(-\bd,-\hat\br)$. 
This is the reciprocity referred to in Section \ref{simo-miso}.

In the case of a point sensor located at $\br_0$ and an incident plane wave,  the measurement data is given by 
\beq
u^{\rm s}(\br_0)
&=&\sum_{j=1}^m \sum_{k=1}^m \lt[ \delta_{j,k} \om^2\nu_{k} +(1-\delta_{j,k})\om^4\nu_{i_j} \nu_{k} \cG(\br_{j},\br_{k})\rt] e^{i\om \bd\cdot \br_{k}}G(\br_0,\br_{j}) \label{dual3}
\eeq
whose right hand side  can also be interpreted as
the scattering amplitude in the direction $-\bd$ when
 a point source is placed at $\br_0$, i.e.
 $A(-\bd, u^{\rm i})$ with $u^{\rm i}(\br)=G(\br,\br_0)$. 
This is the reciprocity referred to in Section \ref{sec:near2}. 

\commentout{

\section{Input-output reciprocity}\label{recipro}

More generally, consider the Helmholtz equation with a source
\[
\Delta u(\br)+\om^2 (1+\nu(\br)) u(\br)= - f(\br)
\]
which can be solved formally as
\[
u=\cG f
\]
by using  the propagator $\cG$. By slight abuse of
notation, we denote the Green function for the propagator still by $\cG$, i.e. 
\[
u(\br)=\int \cG(\br,\br') f(\br') d\br'.
\]
Because the incident wave is governed by
\[
\Delta u^{\rm i}(\br)+\om^2 u^{\rm i}(\br)= - f(\br)
\]
we can write
\beqn
u(\br) &=&-\cG \lt[\Delta u^{\rm i}+\om^2 u^{\rm i}\rt](\br)\\
&=& \cG \lt[-(\Delta+\om^2(1+\nu)) u^{\rm i}+\om^2 \nu u^{\rm i}\rt](\br)\\
&=& u^{\rm i}(\br)+ \om^2 \cG\lt[ \nu u^{\rm i}\rt](\br).
\eeqn
Hence  the scattered wave $u^{\rm s}=u-u^{\rm i}$ can be
expressed as
\beq
\label{exact''}
u^{\rm s}(\br)=\om^2\int \cG(\br,\br') \nu(\br')  u^{\rm i}(\br')d\br'
\eeq
which is the  reciprocal representation  to (\ref{exact'}). 

For point scatterers, (\ref{exact''}) becomes
\beq
\label{dual}
u^{\rm s}(\br)=\om^2\sum_{j=1}^s \nu_{i_j}\cG(\br,\br_{i_j}) 
u^{\rm i}(\br_{i_j}).
\eeq
Substituting the Foldy-Lax equation 
\beq
\cG(\br,\br_j)=G(\br,\br_{i_j})+\om^2\sum_{k\neq j} \nu_{i_k} \cG(\br_{i_j},\br_{i_k}) G(\br,\br_{i_k})
\label{fl3}
\eeq
in (\ref{dual}) we obtain
\beq
u^{\rm s}(\br)&=&\om^2\sum_{j=1}^s \nu_{i_j} u^{\rm i}(\br_{i_j}) G(\br,\br_{i_j}) +\om^4\sum_{j=1}^s\sum_{k\neq j}
\nu_{i_j} \nu_{i_k} \cG(\br_{i_j},\br_{i_k}) u^{\rm i}(\br_{i_j}) G(\br,\br_{i_j})\nn
\eeq
which can be rewritten as
\beq
u^{\rm s}(\br)&=&\sum_{j=1}^s \sum_{k=1}^s  \delta_{j,k} \om^2\nu_{i_k} u^{\rm i}(\br_{i_k}) G(\br,\br_{i_j}) +\om^4\sum_{j=1}^s\sum_{k\neq j}
\nu_{i_j} \nu_{i_k} \cG(\br_{i_j},\br_{i_k}) u^{\rm i}(\br_{i_k})G(\br,\br_{i_j})\nn\\
&=&\sum_{j=1}^s \sum_{k=1}^s \lt[ \delta_{j,k} \om^2\nu_{i_k} +(1-\delta_{j,k})\om^4\nu_{i_j} \nu_{i_k} \cG(\br_{i_j},\br_{i_k})\rt] u^{\rm i}(\br_{i_k})G(\br,\br_{i_j}) \label{dual2}
\eeq
(see \cite{MG, Mar3} for a similar, but slightly  erroneous, expression). 

To obtain the alternative expression for the scattering amplitude,
let $\br \to \infty$ in (\ref{dual2}) and extract the plane wave spectrum by (\ref{somm}) for $d=2$ or (\ref{weyl}) for $d=3$.
We have for the scattering amplitude in the direction $\hat \br$
\beq
A(\hat \br, u^{\rm i})={1\over 4\pi}\sum_{j=1}^s \sum_{k=1}^s \lt[ \delta_{j,k} \om^2\nu_{i_k} +(1-\delta_{j,k})\om^4\nu_{i_j} \nu_{i_k} \cG(\br_{i_j},\br_{i_k})\rt] u^{\rm i}(\br_{i_k})e^{-i\om \hat\br\cdot \br_{i_j}} \label{sa-dual}
\eeq
where $u^{\rm i}$  is not necessarily  a plane wave. 

In the case of a plane wave incidence (\ref{inc}) we observe
the symmetry between the incident and scattered plane waves
in (\ref{sa-dual}). Therefore,  reversing and interchanging roles
of  the incident and scattered waves do not affect the scattering amplitude, i.e. $A(\hat\br,\bd)=A(-\bd,-\hat\br)$. 
This is the reciprocity referred to in Section \ref{simo-miso}.

In the case of a point sensor located at $\br_0$ and an incident plane wave, we have for the measurement 
\beq
u^{\rm s}(\br_0)
&=&\sum_{j=1}^s \sum_{k=1}^s \lt[ \delta_{j,k} \om^2\nu_{i_k} +(1-\delta_{j,k})\om^4\nu_{i_j} \nu_{i_k} \cG(\br_{i_j},\br_{i_k})\rt] e^{i\om \bd\cdot \br_{i_k}}G(\br_0,\br_{i_j}) \label{dual3}
\eeq
which is unchanged by reversing the roles of input and output, i.e. having a point source at $\br_0$
and measuring the scattering amplitude in the direction $-\bd$.
This is the reciprocity referred to in Section \ref{sec:near2}. 
}
\end{appendix}


\begin{thebibliography}{99} 
        
        
\bibitem{BW}
M. Born and E. Wolf, {\em Principles of Optics}, 7-th edition,
Cambridge University Press,  1999. 

\bibitem{BV}
S. Boyd and L. Vandenberghe, {\em Convex Optimization.}
Cambridge University Press, Cambridge, 2004.


\bibitem{BDE}
A.M. Bruckstein, D.L. Donoho and M. Elad,
``From sparse solutions of systems of equations to
sparse modeling of signals," {\em SIAM Rev.}
{\bf 51} (2009), 34-81.

\bibitem{Can}
E. J. Cand\`es, ``The restricted isometry property and its implications for compressed sensing," {\em Compte Rendus de l'Academie des Sciences, Paris, Serie I.} {\bf 346} (2008) 589-592.
 \bibitem{CRT1}
E. J. Cand\`es, J. Romberg and T. Tao, ``Robust undertainty
principles: Exact signal reconstruction from highly incomplete
frequency information," {\em IEEE Trans. Inform. Theory} {\bf 52} (2006), 489-509.


\bibitem{CRT2}
E.J. Cand\`es, J. Romberg and T. Tao, ``Stable signal recovery from incomplete and inaccurate measurements,"
{\em Commun. Pure Appl. Math. } {\bf 59} (2006), 1207Ð23.


\bibitem{CP}
E.J. Cand\`es and Y. Plan, ``Near-ideal model selection by $\ell_1$ minimization," preprint, 2008. 

\bibitem{CT}
E. J. Cand\`es and  T. Tao, `` Decoding by linear programming,'' {\em IEEE Trans. Inform. Theory} {\bf  51} (2005), 4203Ð4215.


\bibitem{CT2}
E. J. Cand\`es and  T. Tao, `` Near-optimal signal recovery from random projections: universal encoding strategies?,'' {\em IEEE Trans. Inform. Theory} {\bf  52} (2006), 54-6-5425. 

\bibitem{Car2}
L.  Carin, D.  Liu and B. Guo,
``In Situ Compressive Sensing for Multi-Static Scattering: 
Imaging and the Restricted Isometry Property", preprint, 2008. 



 


\bibitem{CDS}
S.S. Chen, D.L. Donoho and M.A. Saunders,
``Atomic decomposition by basis pursuit," {\em SIAM Rev.}
{\bf 43} (2001), 129-159.


\bibitem{Che}
M. Cheney, ``The linear sampling method and
MUSIC algorithm," {\em Inverse Problems}
{\bf 17} (2001), 591-596.



\commentout{
\bibitem{CW1}
W.C. Chew and Y.M. Wang, ``An iterative solution of the
two-dimensional electromagnetic inverse scattering problem'',
{\em Int. J. Imaging Syst. Technol. }{\bf 1} (1989), 100-108.

\bibitem{CW2}
W.C. Chew and Y.M. Wang, ``Reconstruction of
two-dimensional permittivity distribution using the distorted
Born iterative method,'' {\em IEEE Trans. Med. Imaging}
{\bf 9} (1990), 218-225. 
}


\bibitem{CCM}
D. Colton, J. Coyle and P. Monk, ``Recent developments
in inverse acoustic scattering theory,'' {\em SIAM Rev.}
{\bf 42} (2000), 369-414.

\bibitem{CK}
D. Colton and R. Kress, {\em Inverse Acoustic and Electromagnetic Scattering Theory.} 2nd edition, Springer, 1998.
 



\bibitem{DM}
W. Dai and O. Milenkovic, ``Subspace pursuit for compressive
sensing: closing the gap between performance and complexity,''
arXiv:0803.0811.

\bibitem{Dau}
I. Daubechies, {\em Ten Lectures on Wavelets.}
SIAM, Philadelphia, 1992.



\bibitem{DE}
D.L. Donoho and M. Elad, ``Optimally
sparse representation in general (nonorthogonal) 
dictionaries via $\ell^1$ minimization,''
{\em Proc. Nat. Acad. Sci. } {\bf 100} (2003) 2197-2202.

\bibitem{DH}
D.L. Donoho and X. Huo, ``Uncertainty principle
and ideal atomic decomposition, '' {\em IEEE Trans. Inform. Theory} {\bf 47} (2001), 2845-2862.



\commentout{
\bibitem{Don2}
D. L. Donoho, ``For most large underdetermined systems of linear equations the minimal  1-norm solution is
also the sparsest solution,"
{\em  Commun. Pure Appl. Math. }
{\bf 59} (2006) 797Ð829. 
}

\bibitem{EHJT}
B.  Efron, T.  Hastie, I.  Johnstone and R. Tibshirani,  
 "Least angle regression". {\em Ann. Statist.}
 {\bf  32} (2004),  407Ð451.

\bibitem{siso-cs}
A. Fannjiang, ``Compressive inverse scattering II. 
SISO measurements with  Born scatterers," to appear. 

\bibitem{subwave-cs}
A. Fannjiang, ``Compressive imaging of subwavelength structures," {\em SIAM J.  Imag. Sci.} {\bf 2} (2009), 1277-1291. 


\bibitem{cs-par}
A. Fannjiang, P. Yan and Thomas Strohmer,
``Compressed remote sensing of sparse objects,"
{\tt arXiv:0904.3994}


\bibitem{GN}
R. Gribonval and M. Nielsen, ``Sparse representation
in unions of bases,'' {\em IEEE Trans. Inform. Theory}
{\bf 49} (2003), 3320-3325.

\bibitem{HN}
G.M. Henkin and R.G. Novikov,
``A multidimensional inverse problem in quantum and
acoustic scattering, "{\em Inverse Problems}
{\bf 4} (1988) 103-121.

\bibitem{Her}
F. J. Herrmann, ``Compressive imaging by wavefield inversion
with group sparsity," preprint, 2009. 
\bibitem{Hoe}
W. Hoeffding, ``Probability inequalities for sums of bounded random variables'', {\em
J.  Amer.  Stat. Assoc.} {\bf 58} (1963) 13Ð30.



\bibitem{Kir}
A. Kirsch, ``The MUSIC-algorithm and the factorization 
method in inverse scattering theory for inhomogeneous media,"
{\em Inverse Problems} {\bf 18} (2002) 1025-1040.

\bibitem{KG}
A. Kirsch and N. Grinsberg, {\em The Factorization Method
for Inverse Problems,} Oxford University Press, Oxford, 2008.



\bibitem{KV}
R. Kohn and M. Vogelius, ``Determining conductivity
by boundary measurements," {\em Comm. Pure Appl. Math.
}{\bf 37} (1984), 113-123.

\commentout{

\bibitem{Lax1} M. Lax, {`` Multiple scattering of waves},''  {\em Rev. Mod.  Phys.} {\bf 23} (1951), 287-310.

\bibitem{Lax2} M. Lax, {``Multiple scattering of waves II.  The effective field in dense
systems},'' {\em  Phys, Rev.} {\bf 85} (1952), 261-269.
}

\bibitem{LP}
P. D. Lax and R. S. Phillips, {\em Scattering Theorey,} Revised Edition. 
Academic Press, San Diego, 1989.


\bibitem{Maj}
A. Majda, ``High frequency asymptotics for the scattering matrix and inverse problem of acoustical scattering," {\em Comm.
Pure Appl. Math.}{\bf 29} (1976) 261-291.

 

\bibitem{Mar3}
E.A. Marengo, "Inverse scattering by compressive sensing and signal subspace methods", {\em  IEEE Workshop on Computational Advances in Multi-Sensor Adaptive Processing (CAMSAP)}, St. Thomas, U.S. Virgin Islands, Dec. 12-14, 2007.

\bibitem{MG} 
 E.A. Marengo and F.K. Gruber, ÒSubspace-based localization and inverse scattering of multiply scattering point 
targetsÓ, {\em EURASIP J. Advances in Signal Processing} {\bf  2007}, Article ID 17342, 16 pages, 2007.

\bibitem{Mar2}
 E.A. Marengo, R.D. Hernandez, Y.R. Citron, F.K. Gruber, M. Zambrano, and H. Lev-Ari, "Compressive sensing for inverse scattering", {\em XXIX URSI General Assembly}, Chicago, Illinois, Aug. 7-16, 2008.

\bibitem{Mel}
R.B. Melrose, {\em Geometric Scattering Theory},
Cambridge University Press, Cambridge, 1995.

\bibitem{Mis}
 M. I. Mishchenko, L. D. Travis, and A. A. Lacis, 
 {\em Multiple Scattering of Light by Particles: Radiative Transfer 
and Coherent Backscattering}  (Cambridge U. Press, Cambridge, UK, 2006).  

\bibitem{Nac}
A. Nachman, ``Reconstruction from boundary measurements,"
{\em Ann. Math.} {\bf 128} (1988), 531-576. 

\bibitem{Nac2}
A. Nachman, ``Global uniqueness for a two-dimensional inverse boundary value problem,"
{\em Ann. Math.} {\bf 143} (1996), 71-96.


 

\bibitem{NTV}
D. Needell, J. A. Tropp, and R. Vershynin, ``Greedy signal recovery review,'' {\em Proc. 42nd Asilomar Conference on Signals, Systems, and Computers}, Pacific Grove, CA, Oct. 2008.


\bibitem{Nov2}
R.G. Novikov, ``The inverse scattering problem on
a fixed energy level for two-dimensional Schr\"odinger 
operator," {\em J. Funct. Anal. }{\bf 103} (1992), 409-463.

\bibitem{Nov3}
R.G. Novikov, ``The inverse scattering problem at fixed energy for the three-dimensional Schršdinger
equationwith an exponentially decreasing potential.'' {\em  Commun. Math. Phys.} {\bf  161} (1994),  569Ð95.


\bibitem{PL}
J. Provost and F. Lesage, ``The application of
compressed sensing for photo-acoustic tomography,"
{\em IEEE Trans. Med. Imag.} {\bf 28} (2009), 585-594.


\bibitem{Ram}
A.G. Ramm, ``Recovery of the potential from fixed energy
scattering data," {\em Inverse Problems} {\bf 4 } (1988), 877-886.

\bibitem{RS}
M. Reed and B. Simon, {\em Methods of Modern Mathematical Physics III.  Scattering Theory}. Academic Press, San Diego, 1979.


\bibitem{SU}
J. Sylvester and G. Uhlmann, `` A global uniqueness
theorem for an inverse problem boundary value problem,"
{\em Ann. Math.} {\bf 125} (1987) 153-169.


\bibitem{The}
C.W. Therrien, {\em  Discrete Random Signals and Statistical Signal Processing}, Englewood Cliffs, NJ: Prentice-
Hall, 1992.

\bibitem{Tib}
R. Tibshirani, ``Regression shrinkage and selection via the lasso,"
{J. Roy. Statist. Soc. Ser. B} {\bf 58} (1996), 267-288.

\bibitem{Tro}
J.A. Tropp, ``Greed is good: algorithmic results for sparse approximation,'' {\em IEEE Trans. Inform. Theory}
{\bf 50} (2004), 2231-2242.

\bibitem{Tro3}
J.A. Tropp, `` Just relax: convex programming methods
for identifying sparse signals in noise,'' {\em IEEE
Trans. Inform. Theory} {\bf 52} (2006), 1030-1051.
``Corrigendum'' {\em IEEE
Trans. Inform. Theory} (2008).

\bibitem{Tro2}
J.A. Tropp, ``On the conditioning of random subdictionaries,''
preprint, 2007. 

\bibitem{TKDA}
L. Tsang, J. A. Kong, K.-H. Ding, and C. O. Ao, 
{\em Scattering of
Electromagnetic Waves: Numerical Simulations}, John Wiley \&
Sons, New York, NY, USA, 2001.


\bibitem{Hul}
H.C. van de Hulst, {\em Light Scattering by Small
Particles.} Dover Publications, New York, 1981.

\bibitem{YL}
J.C. Ye and S. Y. Lee, ``Non-iterative exact inverse scattering using simultanous orthogonal matching pursuit (S-OMP)."
 {\em IEEE Int. Conf. on Acoustics, Speech, and Signal Processing (ICASSP)}, Las Vegas, Nevada, April 2008. 

\end{thebibliography}
\end{document}